\documentclass[12pt]{article}

 \pdfoutput=1
 
\usepackage{natbib}
\usepackage{comment}
\usepackage{hyperref}
\usepackage{mathrsfs}
\usepackage{enumitem}
\usepackage[position=top, labelformat=empty, captionskip=-20pt]{subfig}
\usepackage[font={small,it}]{caption}
\usepackage{amsmath,amsfonts,amsthm,amssymb}
\usepackage{bm,rotating,multirow,dsfont,graphicx}

\newcommand{\blind}{0}
\usepackage{fancyhdr}
\usepackage[ruled,vlined]{algorithm2e}
\usepackage{lscape, pbox, float, hyperref, amsmath, bbm}
\usepackage[usenames, dvipsnames]{color}

\usepackage{hyperref}
\usepackage{url}

\usepackage{lineno}

\usepackage{booktabs}

\newcommand{\bftab}{\fontseries{b}\selectfont}

\renewcommand{\liminf}{\operatornamewithlimits{\textrm{lim\,inf}}}

\usepackage{float}
\restylefloat{table}

\DeclareMathOperator*{\Beta}{Beta}

\newcommand{\R}{\mathbb{R}}

\newcommand{\E}{\mathbb{E}}
\renewcommand{\Pr}{\mathbb{P}}

\newcommand{\N}{\mathbb{N}}

\renewcommand{\P}{\mathcal{P}}

\newtheorem{lemma}{Lemma}
\newtheorem{definition}{Definition}
\newtheorem{theorem}{Theorem}
\newtheorem{coroll}{Corollary}

\usepackage{tikz}
\usetikzlibrary{decorations.pathmorphing} 
\usetikzlibrary{fit}         
\usetikzlibrary{backgrounds} 
\usetikzlibrary{positioning}
\usetikzlibrary{shadows}
\usepgflibrary{shapes}
\tikzstyle{observed}=[circle, thick, minimum size=0.7cm, draw=black!100, fill=black!20]
\tikzstyle{latent}=[circle, thick, minimum size=0.7cm, draw=black!100]
\tikzstyle{plate}=[rectangle, thick, inner sep=0.3cm, draw=black!100]
\tikzstyle{shadeplate}=[rectangle, thick, inner sep=0.3cm, draw=black!100, fill=black!10]

\renewcommand{\epsilon}{\varepsilon}

\newcommand{\x}{\bm{x}}
\newcommand{\btheta}{\boldsymbol{\theta}}
\newcommand{\bbeta}{\boldsymbol{\beta}}

\renewcommand{\liminf}{\operatornamewithlimits{\textrm{lim\,inf}}}

\DeclareMathOperator{\EX}{\mathbb{E}}
\renewcommand{\P}{\textnormal{Pr}}
\renewcommand{\Pr}{\textnormal{Pr}}
\newcommand{\1}{\mathbbm{1}}
\newcommand{\sC}{\mathcal{C}}
\newcommand{\sS}{\mathcal{S}}
\newcommand{\sX}{\mathcal{X}}

\newcommand{\z}{\textbf{z}}
\newtheorem{alg}{Algorithm}

\newtheorem{taggedmodelx}{Model}
\newenvironment{taggedmodel}[1]
 {\renewcommand\thetaggedmodelx{#1}\taggedmodelx}
 {\endtaggedmodelx}
\newtheorem{prop}{Proposition}

\newcommand{\zed}{z_i}

\let\originalleft\left
\let\originalright\right
\renewcommand{\left}{\mathopen{}\mathclose\bgroup\originalleft}
\renewcommand{\right}{\aftergroup\egroup\originalright}

\def\bas#1\eas{\begin{align*}#1\end{align*}}
\def\basn#1\easn{\begin{align}#1\end{align}}

\definecolor{dg}{RGB}{0,170,0}
\definecolor{dred}{RGB}{220,0,0}

\graphicspath{{images/}}

\renewcommand{\liminf}{\operatornamewithlimits{\textrm{lim\,inf}}}

\newcommand{\pkg}[1]{{\normalfont\fontseries{b}\selectfont #1}}
\let\proglang=\textsf

\usepackage{subfiles}
\begin{document}

\def\spacingset#1{\renewcommand{\baselinestretch}%
{#1}\small\normalsize} \spacingset{1}


\if0\blind
{
  \title{\bf Random Partition Models for Microclustering Tasks}
  \author{Brenda Betancourt
  \thanks{BB and RCS acknowledge funding from NSF-1652431 and NSF-1534412.
  GZ acknowledges support from the ERC through StG N-BNP" 306406 and from MIUR, through PRIN Project 2015SNS29B.
  The proposed methodology reflects that of the authors and not of the funding agencies. 
  }
    \\
University of Florida, Department of Statistics,  
\\
       Giacomo Zanella \\
   Bocconi University, Department of Decision Sciences, BIDSA and IGIER, 
\\
and \\
Rebecca C. Steorts\\
Duke University, Department of Statistical Science and Computer Science
}
  \maketitle
} \fi

\if1\blind
{
  \bigskip
  \bigskip
  \bigskip
  \begin{center}
    {\LARGE\bf Random Partition Models for Microclustering Tasks\par}
\end{center}
  \medskip
} \fi

\begin{abstract}
Traditional Bayesian random partition models assume that the size of each cluster grows linearly with the number of data points.
While this is appealing for some applications, this assumption is not appropriate for other tasks such as entity resolution, modeling of sparse networks, and DNA sequencing tasks.
Such applications require models that yield clusters whose sizes grow sublinearly with the total number of data points --- the \emph{microclustering property}.
Motivated by these issues, we propose a general class of random partition models that satisfy the microclustering property with well-characterized theoretical properties. 
Our proposed models overcome major limitations in the existing literature on microclustering models, namely {a} lack of interpretability, identifiability, and full characterization of model asymptotic properties. Crucially, we drop the classical assumption of having an exchangeable sequence of data points, and instead assume an exchangeable sequence of clusters. 
In addition, our framework provides flexibility in terms of the prior distribution of cluster sizes, computational tractability, and applicability to a large number of microclustering tasks.
We establish theoretical properties of the resulting class of priors, where we characterize the asymptotic behavior of the number of clusters and of the proportion of clusters of a given size. Our framework allows a simple and efficient Markov chain Monte Carlo algorithm to perform statistical inference.
We illustrate our proposed methodology on the microclustering task of entity resolution, where we provide a simulation study and real experiments on survey panel data. 
\end{abstract}

\noindent%

\section{Introduction}
Classical prior distributions for Bayesian random partition models assume that the number of data points in each cluster grows linearly with the total number of data points. Traditional examples include finite mixture models, Dirichlet process priors, and 
any model that assumes infinite exchangeability of the observed sequence of the data points. While these models have been successfully used in a wide variety of clustering tasks,
 there are various contexts where the linear growth assumption is undesirable. One such task is that of entity resolution (ER) (record linkage or de-duplication), which is the process of merging together noisy databases in order to remove duplicate entities \citep{winkler_2006, christen_2012}.
Entity resolution can be approached as a clustering task, 
where each entity is implicitly associated with one or more records and the inferential goal is to recover the true latent entities (clusters) that correspond to the observed records (data points)~\citep{copas_1990, liseo_2011, Sadinle14, gutman_2013, steorts15entity, steorts16jasa, zhang2015principled}. 
In this context, the number of data points in each cluster remains small, even for large data sets. 
For example, if there are duplicate records in a database or entities represented in multiple databases, then identifying such duplicates should yield a large number of clusters, many containing only a few data points.
As a result, entity resolution tasks
require models that yield clusters whose sizes grow sublinearly with the total number of data points.
We refer to this as {the \emph{microclustering property},} following the terminology of \citet{miller15microclustering} and \citet{zanella2016microclustering}. For a formal definition of the microclustering property, refer to Definition \ref{def:micro}. For a review of Bayesian methods on entity resolution, refer to \cite{liseo_2013}.

While our proposed methodology for microclustering models is general, we apply it to an entity resolution task due to the growing importance and demand for new methodology in the literature. Entity resolution is not only a crucial task for industrial and social science applications, but is a challenging statistical and computational problem itself because many databases contain errors (noise, lies, omissions, duplications, etc.), and the number of parameters to be estimated grows with the number of records \citep{christen_2012, gutman_2013, Christen08, Larsen02, Larsen05, Sadinle14,  Cohenetal03}. 
Moreover, models for microclustering are relevant to many other applications beyond entity resolution including DNA sequencing \citep{Rashtchian17}, few-to-few matching in language processing \citep{Jitta17}, spatial clustering of settlements in historical research \citep{zanella2015}, and analysis of sparse networks \citep{Bloem-Reddy18}, among others. This motivates us to develop random partition models (RPMs) that satisfy the microclustering problem with well-characterized theoretical properties. 


Our main contribution is proposing a flexible class of RPMs that exhibits the \emph{microclustering property} and providing simple and precise characterization of its theoretical properties. Our model overcomes major limitations in the literature such as lack of
interpretability, identifiability, and full characterization of the model asymptotic properties (see section~\ref{sec:microclust} for a review of related work).
Besides overcoming such limitations, our proposed framework is flexible in terms of being able to control the prior distribution of the cluster sizes. {To achieve this, we propose to use
an exchangeable sequence of clusters,}
where prior information is placed directly on the distribution of the cluster sizes (section~\ref{sec:ESC}).
{Next, we explore the theoretical properties of this corresponding class of priors (section~\ref{sec:asymptotic_analysis}),
describing}
the asymptotic behavior of the number of clusters and the proportion of clusters of a given size. 
Crucially, our model is computationally tractable; and we provide closed-form expressions 
that allow for simple Markov chain Monte Carlo (MCMC) algorithms to perform statistical inference using our proposed framework (section~\ref{sec:Inferences}).
We apply our general methodology to an ER task (section~\ref{sec:RL}), where we illustrate its effectiveness through a simulation study and real experiments based on panel survey data where ground truth is available (section~\ref{sec:sim_both}).
Finally, we provide a discussion and directions for future work (section~\ref{sec:discussion}).

\section{Notation, Background, and Prior Work}
\label{sec:prior-work}

As already mentioned, fundamental issues arise from framing ER as a clustering {task.}
Traditional clustering models assume that the goal is to divide the data into a small number of high-probability clusters. Even if this {assumption} is loosened to allow for a large number of clusters which each have low probability, as in some Bayesian non-parametric (BNP) models, every cluster still has strictly positive probability. This has two consequences, as the number of observations grows. First, every cluster is observed infinitely often. Under exchangeability, observing a data point from a cluster generally makes it more probable to observe more data {points} from that cluster in the future. Second, because every cluster is observed infinitely often, the usual asymptotic theory applies to inferring cluster properties or parameters. When clusters are unique individuals in a population, however, it is not natural to assume that more data {points} 
{imply that eventually an infinite number of records will correspond to this individual.}
Rather, every cluster should be observed a strictly finite number of times, implying that uncertainty about latent individuals does shrink to zero in general. 


In this section, we first present notation used throughout the remainder of the paper, review RPMs, and then propose a generalized microclustering framework for RPMs that can be used to address the issues that arise in ER.

\subsection{Notation}\label{sec:notation}
Consider partitioning $n$ data points $(x_1,\dots,x_n)$ in $K$ clusters $\{\sC_1,\dots,\sC_K\}$, where $K$ is crucially unknown. {That is, we divide the data points into clusters through a partition of the data points' labels $\{1,\dots,n\}$, which we denote by $\Pi_n=\{C_1,\dots,C_K\}$.} For example, the partition $\Pi_n=\{C_1,C_2,C_3\}=\{\{1,3,5\},\{2\},\{4\}\}$
corresponds to the division of data points in the following clusters: $\{\sC_1,\sC_2,\sC_3\}=\{\{x_1,x_3,x_5\},\{x_2\},\{x_4\}\}$. {In addition, we can describe the partition $\Pi_n$
using $n$ cluster allocation variables $(z_1,\dots,z_n)$, where $z_i=k\;$ if and only if $\;i\in C_k$.} In the remainder of the paper, the symbol $[n]$ denotes the set $\{1,\dots,n\}$. For more details, see 
\cite{pitman2006combinatorial}.


\subsection{Background on random partitions}
\label{sec:random_partitions}
{In this section, we provide background on RPMs.}
Let $\Pi_n$ 
be an unknown random partition {of the set} $[n]$ with some prior distribution.
{We assume} the prior distribution of $\Pi_n$ is exchangeable over the set $[n]$, or rather, invariant to permutations of the labels $1,\dots,n$.
This corresponds to assuming that the ordering of the data points $(x_1,\dots,x_n)$ is arbitrary and carries no information about the underlying partition.
It follows by exchangeability over $[n]$ that the prior distribution of $\Pi_n$ can be characterized through its exchangeable partition probability function (EPPF) denoted by $p^{(n)}$ and defined by 
\begin{equation}\label{eq:EPPF_defi}
\P(\Pi_n=\{C_1,\dots,C_K\})=p^{(n)}(S_1,\dots,S_K)\,,
\end{equation}
where $S_j=|C_j|$ is the number of elements in cluster $C_j$.
The  EPPF $p^{(n)}$ is a symmetric function whose input is a so-called partition of $n$ (i.e,\ a collection of positive integers $\{s_1,\dots,s_k\}$ such that $\sum_{j=1}^ks_j=n$) and whose output is a probability value (see section~2.1 of \citealp{pitman2006combinatorial}).
Another consequence of exchangeability over $[n]$ 
is that the distribution of $\Pi_n$ is uniquely determined by the distribution of its cluster sizes $\{S_1,\dots,S_K\}$. 
Specifically, one can generate $\Pi_n$ by first sampling the cluster sizes $\{S_1,\dots,S_K\}$, where $K$ is random, and then defining the cluster allocation variables $(z_1,\dots,z_n)$ by drawing a vector uniformly at random from the set of permutations of 
\begin{equation}\label{eq:indices_order}
(\underbrace{1,\ldots,1}_\text{$S_1$
  times},\underbrace{2,\ldots,2}_\text{$S_2$
  times},\ldots\ldots,\underbrace{K,\ldots,K}_\text{$S_K$ times}).
\end{equation}
Therefore, it is sufficient to specify a model for $\{S_1,\dots,S_K\}$ to uniquely determine the distribution of $\Pi_n$. 
See \citet[Sec.1.2 and Sec.2.1]{pitman2006combinatorial} for further details on the connection between exchangeable random partitions of $[n]$ and the so-called random partitions of $n$.

One important class of exchangeable random partitions of $[n]$ is that of Gibbs partitions \citep{GnedinPitman06}. 
The distribution of Gibbs partitions is characterized by two nonnegative sequences, $\boldsymbol{v}=(v_s)_{s=1}^\infty$ and $\boldsymbol{w}=(w_s)_{s=1}^\infty$, and is denoted by $Gibbs_{[n]}(\boldsymbol{v},\boldsymbol{w})$.
The EPPF of a random partition $\Pi_n\sim Gibbs_{[n]}(\boldsymbol{v},\boldsymbol{w})$ is given by
\[
p^{(n)}(s_1,\dots,s_k)=\dfrac{1}{Z}v_k\prod_{\ell=1}^{k}w_{s_\ell}\,,
\]
where $Z$ is a normalizing constant depending on $n$, $\boldsymbol{v}$ and $\boldsymbol{w}$.
Gibbs partitions allow for a constructive representation, which is crucial to the models proposed in section~\ref{sec:microclust}.
%
Given $\boldsymbol{v}$ and $\boldsymbol{w}$, define two probability distributions on the positive integers $\bm{\kappa}=(\kappa_s)_{s=1}^\infty$ and $\bm{\mu}=(\mu_s)_{s=1}^\infty$ as 
$$
\kappa_s
=\frac{1}{Z_{\boldsymbol{w}}}\frac{w_s}{s!}
\quad\hbox{ and }\quad
\mu_s
=\frac{1}{Z_{\boldsymbol{v}}}\frac{v_s}{s!}
$$
where $Z_{\boldsymbol{w}}=\sum_{\ell=1}^\infty\frac{w_\ell}{\ell!}$ and $Z_{\boldsymbol{v}}=\sum_{\ell=1}^\infty\frac{v_\ell}{\ell!}$ are normalizing constants. Note that $Z_{\boldsymbol{w}}$ and $Z_{\boldsymbol{v}}$ can always be made finite without changing the distribution of $\Pi_n\sim Gibbs_{[n]}(\boldsymbol{v},\boldsymbol{w})$ by assuming $w_s=v_s=0$ for $s> n$.
Given $\bm{\kappa}$ and $\bm{\mu}$, define the random variables $K$ and $\{S_1,\dots,S_K\}$ as 
\begin{align}
K&\sim \bm{\kappa}\label{eq:kolchin1}\\
S_1,\dots S_K|K&\stackrel{iid}\sim \bm{\mu}\,.\label{eq:kolchin2}
\end{align}
Then, the distribution of $\{S_1,\dots,S_K\}$ conditional on $\sum_{j=1}^K S_j=n$ is the same as
the marginal distribution of the cluster sizes of $\Pi_n\sim Gibbs_{[n]}(\boldsymbol{v},\boldsymbol{w})$, see e.g.\ \citet[Thm. 2.1]{pitman2006combinatorial} or Appendix A.1 of \cite{zanella2016microclustering}.
This construction of Gibbs partitions is known as the Kolchin representation \citep{kolchin1971problem}.

A common assumption in the literature is that the observed data points $(x_1,\dots,x_n)$ are the first $n$ elements of an infinite sequence $(x_1,x_2,\dots)$ of exchangeable data points.
In terms of random partitions, this means that $\Pi_n$ can be obtained by restricting some exchangeable random partition of $\mathbb{N}=\{1,2,\dots\}$ to the set $[n]$.
This is equivalent to requiring the two following conditions: (a) each random partition $\Pi_{n}$ is exchangeable over $[n]$; (b) the sequence of random partitions $(\Pi_n)_{n=1}^\infty$ is projective (or Kolmogorov consistent), meaning that $\Pi_{n}$ is equal in distribution to the restriction of $\Pi_{m}$ to $[n]$ for $1\leq n<m$.
We refer to (a) as \emph{finite} exchangeability, to (b) as \emph{projectivity} and to their combination as \emph{infinite} exchangeability.
In this paper, we do not assume infinite exchangeability to hold because this assumption is not appropriate for microclustering  tasks, see sections \ref{sec:microclust} and \ref{sec:previous} for more details.

\subsection{The microclustering property}\label{sec:microclust}
In this section, we present the microclustering property and review previous microclustering models before introducing our proposed methodology. 
It is well known that under the assumption of infinite exchangeability, the size of each cluster grows linearly with the number of data points $n$. However, as discussed in \citet{miller15microclustering, zanella2016microclustering}, this linear growth assumption is inappropriate for some clustering tasks (such as ER). Motivated by this problem, the 
microclustering property {is defined in the aforementioned work as follows: }
\begin{definition}
\label{def:micro}
A sequence of random partitions $(\Pi_n)_{n=1}^\infty$ satisfies the microclustering property if $M_n/n\to0$ in probability as $n\to\infty$, {where $M_n$ is the size of the  largest cluster of the random partition $\Pi_n$.}
\end{definition}

\citet{miller15microclustering, zanella2016microclustering} then proposed two RPMs 
that exhibit the microclustering property, the negative binomial-negative binomial (NBNB) model and the negative binomial-dirichlet (NBD) model. Both models belong to the class of Gibbs partitions and their formulation is based on the Kolchin representation in equations \eqref{eq:kolchin1}-\eqref{eq:kolchin2}. The NBNB model is obtained by assuming that both $\bm{\kappa}$ and $\bm{\mu}$ 
belong to the Negative-Binomial family, i.e.\ $\bm{\kappa}=NegBin(a,q)$ and $\bm{\mu}=NegBin(r,p)$ for some $a,r>0$ and $q,p\in(0,1)$.
The NBD model is {inherently} more flexible as it models $\bm{\mu}$ as a random probability vector with a Dirichlet distribution prior.
The results in \citet{miller15microclustering, zanella2016microclustering} suggest that the sequence of random partitions $(\Pi_n)_{n=1}^\infty$ induced by the NBNB and NBD model exhibit the microclustering property and that these models provide promising prior distributions for ER tasks.
Nevertheless, both the NBNB and NBD models suffer from the following limitations:
\begin{enumerate}
\item \textbf{Interpretability of $\bm{\kappa}$ and $\bm{\mu}$.}
{Turning to the generation process in equations~\eqref{eq:kolchin1} and \eqref{eq:kolchin2},} one may be tempted to interpret $\bm{\kappa}$ as the prior distribution of the number of clusters of $\Pi_n$, which we denote by $K$, and $\bm{\mu}$ as the prior distribution of the size of a randomly chosen cluster from $\Pi_n$, which we denote by $S_j$.
Instead, the actual prior distributions of $K$ and $S_j$ {may be} arbitrarily far from $\bm{\kappa}$ and $\bm{\mu}$.
The reason for such a discrepancy is that the generation process of $\Pi_n$ involves conditioning on $\sum_{j=1}^K S_j=n$, which can dramatically modify the distributions of $K$ and $S_j$. {This makes it difficult to interpret the parameters $\bm{\kappa}$ and $\bm{\mu}$ in the NBNB and NBD models and to specify appropriate prior distributions for these parameters.}
\item \textbf{Statistical Identifiability.} The parameters $\bm{\kappa}$ and $\bm{\mu}$ used in the Kolchin representation are not statistically identifiable. Thus, one can modify $\bm{\kappa}$ and $\bm{\mu}$ without changing the distribution of the resulting random partition $\Pi_n$.
For example, for any $c>0$ consider the probability vectors $\bm{\kappa}^{(c)}=(\kappa^{(c)}_s)_{s=1}^\infty$ and $\bm{\mu}^{(c)}=(\mu^{(c)}_s)_{s=1}^\infty$ defined as $\mu^{(c)}_s=c^s \mu_s/Z_{\bm{\mu}}$ and $\kappa^{(c)}_s=Z_{\bm{\mu}}^s\kappa_s/Z_{\bm{\kappa}}$ for all $s\geq 1$, where $Z_{\bm{\mu}}=\sum_{j=1}^\infty c^j\mu_j$ and $Z_{\bm{\kappa}}=\sum_{j=1}^\infty Z_{\bm{\mu}}^j\kappa_j$  are normalizing constants. 
The random partitions obtained from the Kolchin representation using $(\bm{\kappa}^{(c)},\bm{\mu}^{(c)})$ have the same distribution for any value of $c>0$ (see also \citealp[Sec.1.5]{pitman2006combinatorial}).
\item  \textbf{Asymptotic Properties.} In order to understand the prior assumptions underlying the NBNB and NBD models, one requires a clear characterization of the asymptotic properties of $\Pi_n$ for large $n$, such as the behavior of the number of clusters or the size of a randomly chosen cluster. Currently, there is no such understanding for the NBNB and NBD models.
Moreover, the microclustering property for the NBNB and NBD model is only shown to hold heuristically and proven for {some special cases} \citep{zanella2016microclustering}. 
Such a lack of theoretical results is due to the technical difficulty of working directly with the Kolchin representation, which involves conditioning on the event $\sum_{j=1}^K S_j=n$ whose probability vanishes as $n\to \infty$.
\item \textbf{Parameters varying with $n.$} 
\citet{miller15microclustering, zanella2016microclustering} assume the parameter $\bm{\kappa}$ to be fixed as $n\to\infty$. 
However, this assumption is unrealistic, and in fact, the prior number of clusters should increase with $n$.
In order to specify a reasonable prior in practical microclustering tasks, one needs to choose a different $\bm{\kappa}$ for different values of $n$ (see e.g.\ \citealp[Sec.4]{zanella2016microclustering}).
It is not clear how to choose $\bm{\kappa}$ as $n$ varies and what is the impact of such choices.
\end{enumerate}
We now propose a model that solves the four aforementioned issues. In addition, a more detailed discussion of other work on microclustering models can be found in section \ref{sec:previous}.

\section{Exchangeable Sequences of Clusters (ESC)}\label{sec:ESC}
In this section, we propose a flexible and tractable 
prior distribution for a random partition $\Pi_n$ that is appropriate for microclustering tasks, such as ER. 
As already mentioned, our proposed methods provide a solution to the four issues of microclustering models in the existing literature, while preserving computational tractability by having simple closed-form expressions for the resulting EPPF  (see e.g.\ Corollary \ref{coroll:prediction_rule}).
To achieve these goals, we move away from the traditional assumption of having an exchangeable sequence of data points, $(x_1,x_2,\dots)$ and instead assume that $(\sC_1,\sC_2,\dots)$ is an  exchangeable sequence of clusters with finite sizes.
More formally, we assume that $(\sC_1,\sC_2,\dots)$ is an exchangeable sequence of random elements taking values in $\sS=\cup_{s=1}^\infty \sX^s $, where $\sX$ is a standard Borel space representing the set of possible values of a single data point and $\sX^s$ denotes the $s$-dimensional Cartesian product $\sX\times\cdots\times \sX$.
Denote each $\sC_j$ by 
$$\sC_j=(x^{(j)}_1,\dots,x^{(j)}_{S_j})\in\sS \,,$$
where $S_j$ is the size of $\sC_j$ and each $x^{(j)}_i$ is an element of $\sX$.
By De Finetti's representation theorem, there exist an $\sS$-valued random measure $G$ such that 
$
\sC_1,\sC_2,\dots|G\stackrel{iid}\sim G.
$
Given $G$, denote $\mu_s$ as the probability that a cluster sampled from $G$ has size $s$, and denote $Q_s$ as the cluster distribution conditional on having size equal to $s$.
Denoting $\bm{\mu}=(\mu_s)_{s=1}^\infty$ and $Q=(Q_s)_{s=1}^\infty$ , the random measure $G$ can then be decomposed 
as $G=(\bm{\mu},Q)$,
where sampling $\sC_j\sim G$ is equivalent to first sampling $S_j\sim\bm{\mu}$ and then $\sC_j|S_j\sim Q_{S_j}$.
In specifying a prior distribution for $G$, we assume $\bm{\mu}$ and $Q$ to be independent a priori with distribution $P_{\bm{\mu}}$ and $P_{Q}$, respectively.

Given the sequence of clusters $(\sC_1,\sC_2,\dots)$, the observed data points $(x_1,\dots,x_n)$ arise as the union of a finite number of clusters. More precisely, assume that
\begin{equation}\label{eq:data_from_clusters}
\{x_i\}_{i=1}^n=\bigcup_{j=1}^{K}\sC_j\,,
\end{equation}
where $K$ is the unknown number of clusters.
By doing so, we are implicitly conditioning the sequence $(\sC_1,\sC_2,\dots)$  on the event
\begin{equation}\label{eq:defi_E_n}
E_n=\left\{\hbox{there exists }k \in \N\,\hbox{ such that }\,\sum_{j=1}^k S_j=n \right\}\,.
\end{equation}
Conditional on the occurrence of the event $E_n$,
the random variable $K$ is a function of $(S_1,S_2,\dots)$ defined as the unique positive integer such that $\sum_{j=1}^{K}S_j=n$.
For simplicity throughout the paper, we assume that $\mu_1>0$ almost surely under $P_{\bm{\mu}}$, so that $\P(E_n)>0$ for every $n\geq 1$.

In equation \eqref{eq:data_from_clusters}, we consider the observed data points as an \emph{unordered} set $\{x_i\}_{i=1}^n$.
The ordered vector of data points $(x_1,\dots,x_n)$ is obtained by 
drawing a vector uniformly at random from the set of permutations of
\begin{equation}\label{eq:reordered_points}
(x^{(1)}_1,\dots,x^{(1)}_{S_1},\dots,x^{(K)}_1,\dots,x^{(K)}_{S_{K}})\,.
\end{equation}
This is the same procedure described in equation \eqref{eq:indices_order} in terms of the indices $\{1,\dots,n\}$.


Conditional on the partition of {data} points into clusters, the generation of data points is analogous to classical mixture models, where each cluster $\sC_j=(x^{(j)}_1,\dots,x^{(j)}_{S_j})$ is generated independently of the others according to some distribution $Q_{S_j}$, which depends on the cluster size.  
Therefore, the only difference between classical mixture models and ESC is in the prior distribution on the partition of points into clusters.
Following the notation of section \ref{sec:notation}, we define $\Pi_n=\{C_1,\dots,C_K\}$ to be the partition of $[n]$ induced by the partition of $(x_1,\dots,x_n)$ into the clusters $(\sC_1,\dots,\sC_K)$, meaning that $i\in C_j$ if and only if $x_i\in \sC_j$.
For the ESC models described in this section, the prior distribution of  $\Pi_n$ depends only on $P_{\bm{\mu}}$, which we 
denote by $ESC_{[n]}(P_{\bm{\mu}})$.
A random partition $\Pi_n\sim ESC_{[n]}(P_{\bm{\mu}})$ can be generated, {as defined below.}
\begin{taggedmodel}{$ESC_{[n]}(P_{\bm{\mu}})$}\label{model:esc}
Generate a random partition {of the set} $[n]$ as follows:
\begin{enumerate}[noitemsep,nolistsep]
\item Conditional on $E_n$, sample $\bm{\mu}\sim P_{\bm{\mu}}$ and $S_1,S_2,\dots|\bm{\mu}\stackrel{iid}\sim\bm{\mu}$. 
\item Define $K$ as the unique positive integer such that $\sum_{j=1}^{K}S_j=n$.
\item Define the cluster allocation variables $(z_1,\dots,z_n)$ as a uniformly at random permutation of the vector
\begin{equation}\label{eq:indices_order_2}
(\underbrace{1,\ldots,1}_\text{$S_1$
  times},\underbrace{2,\ldots,2}_\text{$S_2$
  times},\ldots\ldots,\underbrace{K,\ldots,K}_\text{$S_K$ times}).
\end{equation}
\end{enumerate}
\end{taggedmodel}

\subsection{Properties of Exchangeable Sequences of Clusters (ESC)}
\label{sec:asymptotic_analysis}
In this section, we characterize fundamental properties of random partitions using our proposed ESC model in section~\ref{sec:ESC}. {These} properties allow one to perform posterior inference and examine the asymptotic properties of the random partitions of our proposed methodology. Recall that a random partition $\Pi_n\sim ESC_{[n]}(P_{\bm{\mu}})$ is exchangeable over $[n]$, and thus, its distribution can be described through the exchangeable partition probability function (EPPF) defined in equation \eqref{eq:EPPF_defi}. {Below,} we provide explicit expressions for the EPPF of $\Pi_n\sim ESC_{[n]}(P_{\bm{\mu}})$ in Propositions \ref{prop:marginal_EPPF} and \ref{prop:conditional_eppf} and Corollary \ref{coroll:prediction_rule}, which allow one to perform posterior inference using standard computational algorithms (e.g.,\ MCMC).

{First, we provide an explicit expression for the marginal EPPF}. 
\begin{prop}[Marginal EPPF]\label{prop:marginal_EPPF}
The EPPF of a random partition $\Pi_n\sim ESC_{[n]}(P_{\bm{\mu}})$ is given by
\begin{equation}\label{eq:eppf_marginal}
p^{(n)}(s_1,\dots,s_k)
=
\frac{1}{\P(E_n)}
\E_{\bm{\mu}\sim P_{\bm{\mu}}}\left[
\frac{k!}{n!}\prod_{j=1}^k s_j! \mu_{s_j}
\right]
\end{equation}
for any $k\geq 1$ and $(s_1,\dots,s_k)$ satisfying $\sum_{j=1}^ks_j=n$ and $s_j\geq 1$ for all $j$.
The set $E_n$ is defined in equation \eqref{eq:defi_E_n}.
\end{prop}
{We provide a proof of Proposition~\ref{prop:marginal_EPPF} in the {Supplementary Material}, section ~\ref{app:prop1}.}
Depending on the specific choice of $P_{\bm{\mu}}$ the integral in equation \eqref{eq:eppf_marginal} may be more or less tractable.
Thus, it may be more convenient to work with the conditional EPPF, $p^{(n)}(\cdot;\bm{\mu})$, which represents the EPPF of $\Pi_n,$ conditional on the value of the random distribution $\bm{\mu}=(\mu_s)_{s=1}^\infty$. {As such, next, we derive a closed form expression for the conditional EPPF.}

\begin{prop}[Conditional EPPF]\label{prop:conditional_eppf}
Let $\Pi_n\sim ESC_{[n]}(P_{\bm{\mu}})$.
The conditional EPPF $p^{(n)}(\cdot;\bm{\mu})$, defined as $\P(\Pi_n=\{C_1,\dots,C_K\}|\bm{\mu})=p^{(n)}(|C_1|,\dots,|C_K|;\bm{\mu})$, is given by
\begin{equation}\label{eq:eppf_conditional}
p^{(n)}(s_1,\dots,s_k;\bm{\mu})
=
\frac{1}{\P(E_n|\bm{\mu})}
\frac{k!}{n!}\prod_{j=1}^k s_j! \mu_{s_j}\,,
\end{equation}
for any $k\geq 1$ and $(s_1,\dots,s_k)$ satisfying $\sum_{j=1}^ks_j=n$ and $s_j\geq 1$ for all $j$.
\end{prop}
We provide a proof of Proposition~\ref{prop:conditional_eppf} in the {Supplementary Material}, section ~\ref{app:prop2}.
Note that equation \eqref{eq:eppf_conditional} implies that, for fixed $\bm{\mu}$, ESC models fall within the framework of finitely exchangeable Gibbs partitions (section~\ref{sec:random_partitions}). These conditional EPPFs can be used to implement Gibbs samplers and other algorithms 
for prior/posterior sampling for our proposed ESC models.
We now derive the corresponding reallocation probabilities associated with the conditional EPPF $p^{(n)}(\cdot;\bm{\mu})$.

\begin{coroll}[Reallocation probabilities]\label{coroll:prediction_rule}
Let $\z=(z_1,...,z_n)$ be the cluster allocation variables of $\Pi_n\sim ESC_{[n]}(P_{\bm{\mu}})$.
For any $i=1,\dots,n$, the conditional distribution of $z_i$ given $\z_{-i}=\z\backslash z_i$ and $\bm{\mu}$ is
\begin{equation}\label{eq:prediction_rule_conditional}
\P(z_i=j|\z_{-i},\bm{\mu})
\propto
\left\{
\begin{array}{ll}
(s_j+1)
\frac{\mu_{(s_j+1)}}{\mu_{s_j}}& \hbox{if }j=1,\dots,k_{-i},
\\
(k_{-i}+1)\mu_1 & \hbox{if }j=k_{-i}+1\,,
\end{array}
\right.
\end{equation}
where 
$k_{-i}$ is the number of clusters in $\z_{-i}$.
\end{coroll}
{We provide a proof of Corollary~\ref{coroll:prediction_rule} in the {Supplementary Material}, section ~\ref{app:prop2}.}

\subsection{{Asymptotic Behavior of ESC}}
\label{sec:title}
In this section, we characterize the asymptotic behavior of ESC random partitions (Theorems \ref{theorem:Kn}-\ref{theorem:microclustering}), which is critical to understanding the assumptions under our proposed methodology. Specifically, we provide characterizations of the limiting behavior of the number of clusters in $\Pi_n$, give the limiting distribution of the proportion of clusters of a given size, and prove the microclustering property. In the remainder of this section, we assume $\Pi_n$ is a random partition of $[n]$ generated by the ESC model with fixed $\bm{\mu}$, which we denote by $\Pi_n\sim ESC_{[n]}(\bm{\mu})$. Below, we prove a result that concerns the behavior of the number of clusters of $\Pi_n$ as $n$ increases. 


\begin{theorem}[Number of clusters]\label{theorem:Kn}
Assume $\sum_{s=1}s\mu_s< \infty$ and let $K_{n}$ be the number of clusters of $\Pi_n$.
As $n\rightarrow\infty $ it holds
\begin{align}\label{eq:limit_K}
\frac{K_n}{n}&\stackrel{p}\rightarrow \left(\sum_{s=1}^\infty s\mu_s\right)^{-1}\,,
\end{align}
where $\stackrel{p}\rightarrow$ denotes convergence in probability.
\end{theorem}
{We provide a proof of Theorem \ref{theorem:Kn} in the {Supplementary Material}, section~\ref{app:thm}.}

Theorem \ref{theorem:Kn} implies that, if the mean of $\bm{\mu}$ is finite (i.e.\ $\sum_{s=1}^\infty s\mu_s< \infty$) the number of clusters of $\Pi_n$ grows linearly with the number of data points $n$.
This corresponds to the behavior empirically observed, e.g., in ER tasks. 
Next, we provide a characterization of the limiting distribution of the number of clusters of $\Pi_n$ with a given size.
\begin{theorem}[Proportion of clusters of given size] \label{theorem:fof}
Assume 
$\sum_{s=1}s\mu_s< \infty$.
\begin{enumerate}
\item[(a)] Let $M_{s,n}$ be the number of clusters of size $s$ in $\Pi_n$. 
As $n\rightarrow\infty $
$$
\frac{M_{s,n}}{n}\stackrel{p}\rightarrow \frac{\mu_s}{\sum_{\ell=1}^{\infty}\ell\mu_\ell}\,.
$$
\item[(b)] 
The size $S_j$ of a cluster chosen uniformly at random from the clusters of $\Pi_n$ converges in distribution to $\bm{\mu}$ as $n\rightarrow\infty $.
\end{enumerate}
\end{theorem}
{We provide a proof of Theorem \ref{theorem:fof} in the {Supplementary Material}, section~\ref{app:thm} .}

Compared to Theorem \ref{theorem:Kn}, the latter result provides a more refined characterization of the structure of a typical draw of ESC random partitions. In particular, Theorem \ref{theorem:fof} implies that, asymptotically, the distribution of the size of a randomly chosen cluster from $\Pi_n$ coincides with $\bm{\mu}$, which provides a natural interpretation for such a parameter.
This allows one to specify appropriate prior distributions for $\Pi_n$, by incorporating into $P_{\bm{\mu}}$ any knowledge about the expected behavior of the sizes of the clusters of $\Pi_n$, and also 
to extract meaningful information from the posterior distribution of $\bm{\mu}$.
As discussed in section \ref{sec:microclust}, a result analogous to Theorem \ref{theorem:fof} does not hold for models based on the Kolchin representation. Thus, interpreting $\bm{\mu}$ as the distribution of the size of a randomly chosen cluster from $\Pi_n$ is not appropriate for models such as NBNB or NBD.

Finally, we provide a proof of the microclustering property for ESC random partitions.
This result requires $\bm{\mu}$ to have finite mean, i.e.\ $\sum_{s=1}^\infty s\mu_s< \infty$.
\begin{theorem}[Microclustering]\label{theorem:microclustering}
Let $M_{n}$ be the size of the largest cluster in $\Pi_n$ and assume $\sum_{s=1}^\infty s\mu_s< \infty$.
As $n\rightarrow\infty $
$$
\frac{M_{n}}{n}\stackrel{p}\rightarrow 0\,.
$$
\end{theorem}
The proof of Theorem \ref{theorem:microclustering} can be found in section \ref{app:thm} of the supplement.

\subsection{Model specification and inferences}\label{sec:Inferences}
In this section, we specify two instances of ESC models that will be used both in a simulation study and a real application.
In addition, we provide simple-to-implement MCMC schemes to perform posterior inferences.
The two models, which we refer to as ESC-NB and ESC-D, differ in the way the prior distribution $P_{\bm{\mu}}$ is specified. Recall that $P_{\bm{\mu}}$ is a distribution over probability distributions $\bm{\mu}=(\mu_s)_{s=1}^\infty$ on the positive integers. 

\subsubsection{ESC-NB model}\label{sec:ESCNB}
First, we model $\bm{\mu}=(\mu_s)_{s=1}^\infty$ as a Negative Binomial distribution truncated on $\{1,2,\dots\}$, denoted by $\bm{\mu}=NegBin(r,p).$ {This implies that for all $s=1,2,\dots$}
\begin{equation}\label{eq:mu_esc_NB}
\mu_s(r,p)=\gamma\frac{\Gamma(s+r)p^s }{\Gamma(r)s!}\,,
\end{equation}
where $\gamma=\dfrac{(1-p)^{r}}{1-(1-p)^r}$.
{Here, $r>0$ and $p\in(0,1)$ are unknown parameters with prior distributions}
$$r \sim \text{Gamma}(\eta_r,s_r)\; \text{and} \; p \sim \Beta(u_p,v_p),$$
where the {fixed hyperparameters} $\eta_r$, $s_r$, $u_p$ and $v_p$ are chosen to reflect the prior expectations on the distribution of the cluster sizes $S_j$ for the application under consideration (see e.g.\ section~\ref{sec:experiments}).
We refer to the resulting model as ESC-NB.

{When the ESC-NB model is combined with a likelihood function} $\Pr(\x|\Pi_n)$, where $\x=(x_1,\dots,x_n)$, (e.g.\ the one specified in section \ref{sec:RL}), {this combination leads to a}  joint posterior distribution 
$\Pr(\Pi_n,r,p|\x)$.
Using the results of Corollary \ref{coroll:prediction_rule}, it is straightforward to sample from $\Pr(\Pi_n,r,p|\x)$ using an MCMC scheme that iterates the following steps:
\begin{enumerate}
\item Update $(r,p)|\Pi_n,\x$ with a MCMC kernel that is invariant with respect to the following conditional distribution
\begin{align}\label{eq:cond_par_ESC_NB}
\Pr(r,p|\Pi_n,\x) &
\propto
r^{\eta_r-1}e^{-\frac{r}{s_r}}
p^{n+ u_p-1}(1-p)^{ v_p-1}
\gamma^K\prod_{j=1}^K\frac{\Gamma(S_j+r)}{\Gamma(r)}\,.
\end{align}
where $K$ is the number of clusters in $\Pi_n$ and $S_j$ the size of the $j$-th cluster.
\item Update $\Pi_n|\x,r,p$ with a MCMC kernel (e.g.\ Gibbs Sampling) using the following full conditionals for the allocation variables
\begin{equation}\label{eq:full_cond_ESCNB}
\P(z_i=j|\z_{-i},\x,r,p)
\propto
\Pr(\x|\z_{-i},z_i=j)\times
\left\{
\begin{array}{ll}
S_j+r & \hbox{if }j=1,\dots,K_{-i},
\\
(K_{-i}+1)\gamma r & \hbox{if }j=K_{-i}+1\,
\end{array}
\right.
\end{equation}
where $K_{-i}$ is the number of clusters in $\z_{-i}$.
\end{enumerate}
See section \ref{supp:samplers} of the supplement for the derivation of equations \eqref{eq:cond_par_ESC_NB} and \eqref{eq:full_cond_ESCNB}.
In our simulation study and real data applications, we use slice sampling \citep{neal2003slice} to perform Step 1 and the \emph{chaperones algorithm} described in \cite{miller15microclustering} to perform Step 2. 

\subsubsection{ESC-D model}\label{sec:ESCD}
The assumption underlying the ESC-NB model that $\bm{\mu}$ coincides with a Negative Binomial distribution is potentially restrictive.
A more flexible choice is to model directly $\bm{\mu}=(\mu_s)_{s=1}^\infty$ as a random distribution and assign a prior distribution to it, such as a Dirichlet distribution $\bm{\mu}\sim Dir(\alpha,\bm{\mu}^{(0)})$, where $\bm{\mu}^{(0)}=(\mu^{(0)}_s)_{s=1}^\infty$ is a sequence of non-negative numbers satisfying $\sum_{s=1}^\infty\mu^{(0)}_s=1$.
We impose a parametric form on $\bm{\mu}^{(0)}$ in an analogous way to the ESC-NB model, i.e.,\ $\bm{\mu}^{(0)}$ is truncated Negative Binomial distribution with
\begin{equation}\label{eq:mu0_esc}
\mu^{(0)}_s(r,p)=\gamma\frac{\Gamma(s+r)p^s }{\Gamma(r)s!}\,,
\end{equation}
where $r$ and $p$ are unknown parameters {with prior distributions}
$$r \sim \text{Gamma}(\eta_r,s_r)\; \text{and} \; p \sim \Beta(u_p,v_p),$$
We refer to the resulting model as ESC-D.
{Under such model, $\bm{\mu}$ is allowed to be any distribution on $\{1,2,\dots\}$ a priori, and $\bm{\mu}^{(0)}$ represents our {a priori expectation about $\bm{\mu}$.}
It is important to note that $\alpha$ represents the degree of confidence we have in the fact that $\bm{\mu}$ is close to $\bm{\mu}^{(0)}$, and {changing 
$\alpha$ allows} one to enforce more or less smoothness in the the distribution of $\bm{\mu}$. Thus, as $\alpha$ increases, then $\bm{\mu}$ becomes more smooth; and in the 
limit $\alpha\to\infty$, the ESC-D model coincides with the ESC-NB one.}

When the ESC-D model is combined with a likelihood function $\Pr(\x|\Pi_n)$, {this combination leads to joint posterior distribution} over $\Pi_n$, $\bm{\mu}$, $r$ and $p$.
In order to sample from the posterior distribution of $\Pr(\Pi_n,\bm{\mu},r,p|\x)$ we use a scheme that iterates over the following steps:
\begin{enumerate}
\item Update $(r,p)|\Pi_n,\x$ with a MCMC kernel that is invariant with respect to the following conditional distribution 
\begin{align}
\Pr(r,p|\Pi_n,\x)
&\propto
r^{\eta_r-1}e^{-\frac{r}{ s_r}}
p^{ u_p-1}(1-p)^{ v_p-1}\prod_{s=1}^{M_n} 
\frac{\Gamma(M_{s,n}+\alpha\,\mu^{(0)}_s)}{\Gamma(\alpha\,\mu^{(0)}_s)}\,,
\label{eq:ESC_D_parameters}
\end{align}
where $K$ is the number of clusters in $\Pi_n$, $M_n$ the maximum cluster size and $M_{s,n}$ the number of clusters of size $s$.
\item Sample $(\mu_1,\dots,\mu_n,1-\sum_{i=1}^n\mu_i)|\Pi_n,r,p\sim$\\ 
\,\hbox{} \hfill$ \hbox{Dir}(\alpha\mu_1^{(0)}+ M_{1,n},\dots,\alpha\mu_n^{(0)}+M_{n,n},\alpha(1-\sum_{i=1}^n\mu_{i}^{(0)}))$.
\item Update $\Pi_n|\textbf{x},\bm{\mu},r,p$ 
 using some MCMC kernel (e.g.\ Gibbs Sampling) using the following full conditionals for the allocation variables
 \begin{equation*}
\P(z_i=j|\z_{-i},\bm{\mu})
\propto
\Pr(\x|\z_{-i},z_i=j)\times\left\{
\begin{array}{ll}
(S_j+1)
\frac{\mu_{(S_j+1)}}{\mu_{S_j}}& \hbox{if }j=1,\dots,K_{-i},
\\
(K_{-i}+1)\mu_1 & \hbox{if }j=K_{-i}+1\,.
\end{array}
\right.
\end{equation*}
\end{enumerate}
{In steps 1-3, we utilize} a partially-collapsed Metropolis-within-Gibbs sampler \citep{van2015metropolis}.
The sampler collapses $\bm{\mu}$ when sampling $(r,p)$ in Step 1 to overcome the strong dependence between $\bm{\mu}$ and $(r,p)$.
This is possible because, thanks to the Dirichlet-Multinomial conjugacy, the marginal EPPF describing $\Pr(\Pi_n|r,p)$ is tractable and one can obtain the analytic expression in equation~\eqref{eq:ESC_D_parameters} by computing explicitly the integral in equation \eqref{eq:eppf_marginal}.
See section \ref{supp:samplers} of the supplement for the derivation of equation \eqref{eq:ESC_D_parameters}.
At the same time, the scheme imputes the vector $\bm{\mu}$ in Step 2 and explicitly conditions on it when sampling $\Pi_n$ in Step 3. This allows to exploit the higher degree of tractability and conditional independence of $\Pi_n|\bm{\mu},r,p$ compared to $\Pi_n|r,p$, which facilitates the MCMC update of $\Pi_n$ in Step 3. 
Note that in the actual implementation it is not needed to generate the whole vector $(\mu_1,\dots,\mu_n,1-\sum_{i=1}^n\mu_i)$ but only the first $m$ components for some $m$ larger than $M_n$ and potentially impute more components if $M_n$ exceeds $m$ during Step 3.

\subsubsection{Generating samples from $ESC$ models}

In this section, we discuss how to generate samples from the random partition $\Pi_n\sim ESC_{[n]}(P_{\bm{\mu}})$.
A simple and natural approach is to use a rejection sampler that proceeds as follows:
\begin{enumerate}
\item Sample $\bm{\mu}\sim P_{\bm{\mu}}$. Then sample
 $S_1,\dots,S_K|\bm{\mu}\stackrel{iid}\sim\bm{\mu}$ until one obtains the first value of $K$ such that $\sum_{j=1}^KS_j\geq n$.
\item 
If $\sum_{j=1}^KS_j= n$, obtain the cluster allocation variables $(z_1,\dots,z_n)$ by uniformly permuting the allocations given in \eqref{eq:indices_order_2}; otherwise return to Step 1 and repeat.
\end{enumerate}
The proposed rejection sampler only requires the user to be able to generate samples from $\bm{\mu}\sim P_{\bm{\mu}}$ and $S_j\sim\bm{\mu}$, which is straightforward for the ESC-NB and ESC-D models.
Moreover, the average number of iterations of Steps 1 and 2 required by the algorithm for each accepted sample is $Pr(E_n)^{-1}$, where $E_n$ is the event defined in \eqref{eq:defi_E_n}.
Assuming $\mu_1>0$  and $\sum_{s=1}s\mu_s<\infty$ almost surely under $P_{\bm{\mu}}$,
Theorem \ref{thm:renewal_theorem} from the Supplement implies that
$Pr(E_n)$ converges to $\E_{\bm{\mu}\sim P_{\bm{\mu}}}\left[(\sum_{s=1}s\mu_s)^{-1}\right]$
 as $n\to\infty$, which is a constant strictly greater than $0$. 
It follows that the average number of iterations required for each accepted sample, $Pr(E_n)^{-1}$, remains bounded as $n\to\infty$, making the rejection sampler computationally feasible when $n$ is large. 
In addition, section \ref{sec:IS} of the supplement describes a more sophisticated importance sampler that allows one to sample efficiently from $\Pi_n\sim ESC_{[n]}(P_{\bm{\mu}})$ when $\E_{\bm{\mu}\sim P_{\bm{\mu}}}\left[(\sum_{s=1}s\mu_s)^{-1}\right]$ is very small.

\subsection{Projectivity and prior work on microclustering}\label{sec:previous}

In this section, we review prior work in the literature on microclustering, focusing on its connections to projectivity. 
As mentioned in section \ref{sec:microclust}, the microclustering property is incompatible with infinite exchangeability.
Thus, any non-trivial microclustering model must either sacrifice finite exchangeability or projectivity (see section \ref{sec:random_partitions}).

Similarly to \citet{miller15microclustering, zanella2016microclustering}, our proposed model sacrifices projectivity. 
Indeed finite exchangeability is a highly appropriate assumption in applications such as ER, where the order of entries in datasets is typically arbitrary and carries no information on the underlying linkage structure.
Sacrificing projectivity can create issues when trying to extrapolate inferences obtained from a given sample to the whole population. However, while being a non-trivial limitation, the latter is not directly relevant in ER contexts, where the typical inferential goal is to resolve entities in a given dataset.

In addition, because ESC models are built on infinitely exchangeable sequences of clusters, one can naturally perform prediction and extrapolate inferences to the population level in terms of future clusters -- rather than future observations -- using the posterior predictive distribution of a new cluster given the observed ones.
In doing so, one is implicitly assuming that observations arose as a random subset of clusters -- rather than a random subset of data points -- 
taken from a larger population.
The validity of this assumption depends on the application and, conditional on that, ESC models may or may not be appropriate to use for inference at the population level.

We now turn to reviewing related work on microclustering in the literature.
\cite{DiBenedetto2017} have developed a model for microclustering based on completely random measures, providing a careful characterization of its theoretical properties.
More specifically, the authors are motivated by data sets with a temporal component (e.g.\ arrival times), and thus, their proposed model preserves projectivity and abandons finite exchangeability. Thus, their model is designed for contexts where the order of data points is informative about the latent partition, which differs from our proposed approach. In addition, the marginal distribution of the random partition they propose is not fully tractable \citep[Sec.4.2.2]{DiBenedetto2017}, which creates computational challenges when performing posterior inference. In other work, \cite{klami2016probabilistic, jitta2018controlling} have proposed microclustering model in the context of finite mixture models. Their model shares some similarity with the ESC one in the sense of  directly modeling the distribution of cluster sizes, though they assume the number of clusters $K$ to be known and impose hard constraints on cluster sizes (e.g. by assuming clusters to be i.i.d.\ from a known distribution, rather than exchangeable). See also \citet{silverman2017bayesian} for a decision-theoretic approach to impose constraints on cluster sizes in the context of finite mixture models, which could potentially be extended to accommodate microclustering.

Non-projective random partition models have also been studied in \cite{zhou2017frequency}, where the authors focus on modeling random partitions that depend on both the sample size and the population size. Their models exhibit different asymptotic behaviors than ours (e.g.\ the number of clusters of a given size remains bounded as $n\to\infty$), and it is not clear if these models satisfy microclustering. Turning to projectivity and invariance under subsampling, this has been studied both in the random partitions literature \citep{kingman1978representation,gnedin2009characterizations} and, more extensively, in the random networks one \citep{shalizi2013consistency,orbanz2017subsampling}.
In particular, there have been proposals on how to relax exchangeability to obtain models for sparse networks \citep{crane2018edge,cai2016edge}.
These ideas have been applied to microclustering by \citet{Bloem-Reddy18}, focusing on random partition models with power-law distributions of cluster sizes.


Finally, \cite{steorts2017performance} and \cite{dunson17} provide performance bounds for ER tasks using microclustering models. 
Given the number of features, categories within the features and the level of noise, \cite{steorts2017performance} provide quantitative statements on what is the largest number of entities one may hope to resolve assuming a general class of models. 
The authors achieve this by 
providing lower bounds on the minimum probability of misclassifying latent entities.
\cite{dunson17} show that unless the number of features grows with the number of data points (records), entity resolution is infeasible in certain contexts.
This agrees with empirical evidence from the entity resolution literature and is related to the separation between entities going to zero as the number of entities increases.
Given the work of \cite{dunson17},  
 it seems that studying more closely the regime where the number of features grows with the number of data points may provide insight on how much information one needs to collect for a given entity resolution problem. 
For example, the authors suggest that a logarithmic growth in the number of features may be sufficient to achieve accurate entity resolution, and if substantiated, this would be a rather encouraging scenario. 
Both aforementioned papers have similar goals and it would be interesting to combine and expand their results, seeing what guidance such theory could provide to a user.


\section{Entity Resolution Model for Categorical Data} \label{sec:RL}

In this section, we describe the data generation process in the context of ER using a fully Bayesian generative process \citep{steorts16jasa, liseo_2011}. 
In many ER tasks, unique identifiers (e.g. social security numbers) and other personal identifying information (e.g. full name, address) are not easily accessible due to privacy concerns \citep{christen_2012, winkler_2006}. In such cases, we need to rely on information of commonly available categorical fields (e.g. gender, date of birth), to cluster records that belong to the same entity.  

{Before providing the model, we first define notation and assumptions used throughout the remainder of the paper. Assume} the observed data $\x = (x_1,\dots,x_n)$ consists of $n$ records, and each record $x_i$ contains $L$ fields $(x_{i\ell})_{\ell=1}^L$. Assume that the fields correspond to categorical features where $D_\ell$ denotes the number of categories in the $\ell$-th field. {Following the framework proposed by \cite{steorts16jasa},} we assume that {fields within a cluster are conditionally independent. In addition, there are the following field specific parameters:}
\begin{itemize}
\item a probability $\beta_\ell\in(0,1),$ {reflecting the distortion at the field level,}
and
\item a density vector $\btheta_{\ell}=(\theta_{\ell d})_{d=1}^{D_\ell}\in[0,1]^{D_\ell}$, {characterizing} the distribution of categories {within each} field, where $\sum_{d=1}^{D_\ell}\theta_{\ell d}=1$. 
\end{itemize}
Let $y_{j\ell}$ represent the {latent entity}
associated to cluster $C_j \in \Pi_n$ for $j=1,\ldots, K$, and  $\zeta(\Pi_{n}, i)$ {represent} a function that maps record $i$ to its latent cluster assignment $z_i$ according to 
$\Pi_{n}$. {Then assuming a spike-and-slab distribution with mixture weight $\beta_\ell$ for the likelihood function in equation \eqref{eq:x_giv_y}, we write the microclustering model for ER as}
\begin{align}
x_{i\ell} | y_{j\ell}, \zed, \btheta_{\ell}, \beta_{\ell}  &\stackrel{ind}\sim \beta_{\ell}\btheta_{\ell} + (1-\beta_{\ell})\delta_{y_{z_i\ell}} \label{eq:x_giv_y}\\
y_{j\ell} & \stackrel{ind}{\sim} \btheta_{\ell}  \label{eq:y_gen} \\
\zed & = \zeta(\Pi_{n}, i)  \\
\Pi_n & \sim ESC_{[n]}(\bm{\mu}),
\end{align}
where $\delta_{y}$ is the Dirac-delta function at $y$, and $\btheta_{\ell}$ is fixed and assumed to be the empirical distribution of the data.  By integrating out $y_{j\ell}$ from equations \eqref{eq:x_giv_y} and \eqref{eq:y_gen}, it follows that the likelihood function is
\begin{equation}\label{eq:likelihood_function}
\Pr(\x|\Pi_n,\bbeta,\btheta)
=
\left(
\prod_{\ell=1}^L
\prod_{i=1}^n
\beta_\ell
\theta_{\ell x_{i\ell}}
\right)
\prod_{j=1}^K
\prod_{\ell=1}^L
f(\x_{j\ell},\beta_\ell,\btheta_\ell)
\,,
\end{equation}
where $\x_{j\ell}=\{x_{i\ell}\,:\,i\in C_j\}$. Moreover, 
\begin{align}\label{eq:f}
f(\x_{j\ell},\beta_\ell,\btheta_\ell)
&=
1-\sum_{i=1}^{m^{(j)}}\theta_{\ell x^{(j)}_{i\ell}}
+\sum_{i=1}^{m^{(j)}}\theta_{\ell x^{(j)}_{i\ell}}
\left(
\frac{\beta_\ell\theta_{\ell x^{(j)}_{i\ell}}+(1-\beta_\ell)}{\beta_\ell\theta_{\ell x^{(j)}_{i\ell}}}
\right)^{q^{(j)}_{i\ell}}
\,.
\end{align}
Note that $x^{(j)}_{1\ell},\dots,x^{(j)}_{m^{(j)}\ell}$ represents the collection of unique values in $\x_{j\ell}$, and $q^{(j)}_{1\ell},\dots,q^{(j)}_{m^{(j)}\ell}$  represents the corresponding frequencies such that for each $i\in\{1,\dots,m^{(j)}\}$ the value $x^{(j)}_{i\ell}$ appears exactly $q^{(j)}_{i\ell}$ times in $\x_{j\ell}$ (see additional details in section \ref{supp:likelihood} of the supplement). For choices of the fixed parameters of the model, we refer to \cite{steorts15entity} for practical guidance and extensive simulation studies.


\section{Simulation Study and Real Data Applications}\label{sec:sim_both}

In this section, we explore our proposed methods on both simulated and real data sets for entity resolution tasks. Specifically, we explore the behavior of the proposed ESC-NB and ESC-D partition priors, comparing their performance to two traditional partition models based on Dirichlet Process (DP) and Pitman-Yor (PY) process mixtures \citep{sethuraman94constructive, ishwaran03generalized}. The aim of the simulation study is to explore the impact of the choice of prior model for the random partition on ER performance for varying levels of signal-to-noise ratio and varying distributions of cluster sizes for the true data-generating partition. The aim of the real data experiments is to evaluate how well the proposed models recover the true partition in more realistic scenarios where noise levels are unknown and the likelihood is potentially misspecified.

\subsection{Simulation Study}\label{sec:sim_study}

In this section, we simulated data according to the ER model described in section \ref{sec:RL} with $L=5$ fields, $D_\ell=10$ categories per field, a uniform distribution $\btheta_\ell$ on $\{1,\dots,10\}$, and a distortion parameter $\beta$ varying in $\{0.01,0.05,0.1\}$.
For each value of the distortion parameter, we consider five different 
data-generating partitions with different combinations of cluster sizes.
All distributions of cluster sizes corresponds to $K=200$ clusters (e.g.\ the first partition contains 50 clusters of size 1, 2, 3 and 4, respectively) and vary in terms of average size of clusters and shape of the resulting distribution, as shown in Figure \ref{fig:scenarios}. 

We assume the parameters $\btheta_\ell$ and $\beta$ to be known, keeping them fixed at their data-generating values, and focus on the posterior distribution of the random partition, $\Pr(\Pi_n|\x)$.
We also experimented with estimating $\btheta_\ell$ with the data empirical distribution and assigning a prior to $\beta$.
In most cases, the resulting posterior distribution of $\beta$ was concentrated around the data-generating value and the overall results were qualitatively similar to the ones reported below.
For each combination of the distortion level and data-generating partition, we compare the four models in terms of the posterior False Negative Rates (FNR) and False Positive Rates (FPR) that they achieve, which are estimated using the MCMC algorithms described in sections \ref{sec:ESCNB}-\ref{sec:ESCD} (see section \ref{supp:convergence} of the supplement for implementation details). 

\begin{figure}[t!]
\includegraphics[width=\textwidth]{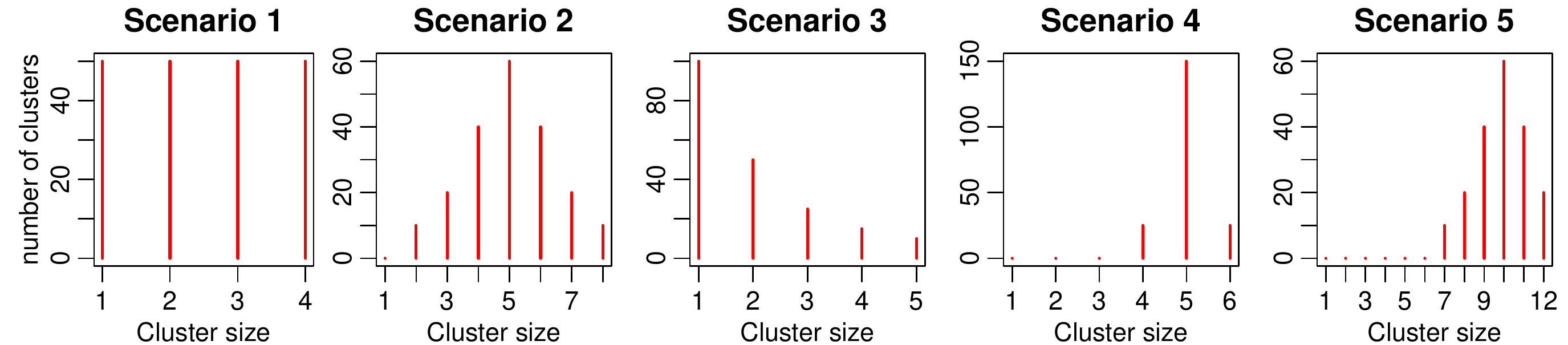}
\caption{Number of clusters of each size for the five data-generating partitions considered in the simulation study of section \ref{sec:sim_study}.} \label{fig:scenarios}
\end{figure}

The results, reported in Table \ref{table:rates_simulations}, are insightful in many ways. First, observe that the ESC-NB and ESC-D outperform the DP and PYP models  in most situations, especially in terms of significantly reducing the FNR while maintaining comparable FDR. The extent of the improvements depends on both distortion level and distribution of cluster sizes. Next, Figure \ref{fig:boxplots_simulations} reports the posterior distribution of the number of clusters of each size (black box-plots) versus the number of clusters of each size in the true data-generating partition (red dots) for $\beta=0.05$. In particular, Figure \ref{fig:boxplots_simulations} suggests that DP ad PYP are flexible enough to accommodate the smooth ``geometric-like" distribution of cluster sizes of Scenario 3, where they achieve FNR and FDR comparable to ESC-NB and ESC-D. On the other hand, the DP and PYP models struggle to accommodate other distributions of cluster sizes such as the ones in Scenarios 1, 2, 4 and 5. Figures analogous to Figure \ref{fig:boxplots_simulations} with $\beta=0.01$ and $\beta=0.10$ can be found in the supplement.

\begin{table}[h!]
{\renewcommand{\arraystretch}{1.2}
\begin{footnotesize}
\centering
\begin{tabular}{c|c|cc|cc|cc|cc|cc}
 & & \multicolumn{2}{c|}{Scenario 1} & \multicolumn{2}{c|}{Scenario 2}  & \multicolumn{2}{c|}{Scenario 3}  & \multicolumn{2}{c|}{Scenario 4}  & \multicolumn{2}{c}{Scenario 5}  \\ 
  \hline
 & Model & FNR & FDR & FNR & FDR & FNR & FDR & FNR & FDR & FNR & FDR \\ 
  \hline
\parbox[t]{2mm}{\multirow{3}{*}{\rotatebox[origin=c]{90}{$\beta=0.01$}}} 
& DP & 6.2 & 1.1 & 2.1 &  0.2 & 6.4 &  1.3 & 2.8 & 0.7 & 1.0 &  0.3 \\ 
& PY & 6.1 & 1.1 & 2.2 &  0.2 & 6.5 &  1.3 & 2.8 & 0.7 & 1.0 &   0.3 \\ 
& ESCNB & 4.3 & 1.3 & 0.8& 0.4 & 6.4 & 1.4 & 1.1 & 0.7 &  0.3 &  0.3 \\ 
& ESCD & 2.9 & 1.2 &  0.7 & 0.4 &  5.7 & 1.6 &  0.3 &  0.3 &  0.3 &  0.3 \\ 
  \hline
    \hline
\parbox[t]{2mm}{\multirow{3}{*}{\rotatebox[origin=c]{90}{$\beta=0.05$}}} 
& DP & 11.7 & 6.4 & 9.2 &  4.1 & 14.6 &  6.9 & 8.7 & 4.3 & 6.7&  3.3\\ 
& PY & 11.9 & 6.4 & 9.3 &  4.1 & 14.5 & 7.0 & 8.7 & 4.3 & 6.7 &  3.3 \\ 
& ESCNB & 9.0 & 6.4 & 6.4 & 5.1 & 13.7 & 7.4 & 5.9 & 5.0 &  4.3 & 4.2 \\ 
& ESCD &  8.0 &  4.4 &  5.9 & 5.6 & 12.8 & 7.7&  3.0 &  2.3 & 4.7 & 3.7 \\ 
  \hline
    \hline
\parbox[t]{2mm}{\multirow{3}{*}{\rotatebox[origin=c]{90}{$\beta=0.1$}}} 
& DP & 27.2 & 16.3 & 17.4 & 11.3 &  33.9 & 16.2 & 22.9 & 12.2 & 18 &  11.1 \\ 
& PY & 27.5 & 16.3 & 17.4 &  11.2 & 34.1 & 16.1 & 23.0 & 12.2 & 17.9 &  11.1 \\ 
& ESCNB & 24.3 & 16.3 & 13.3 & 12.8 &  33.9 & 15.9 & 16.9 & 13.9 &  13.8 & 13.7 \\ 
& ESCD &  21.7 &  14.0 &  12.5 & 12.9 & 34.7 &  15.4 &  11.3 & 8.9 & 14.3 & 11.4 \\ 
\end{tabular}
\caption{Posterior FNR and FDR (in percentages) for the simulation study of section \ref{sec:sim_study}, with different noise levels ($\beta\in\{0.01,0.05,0.1\}$), distributions of cluster sizes (Scenarios 1-5) and models for the prior distribution over partitions (DP/PYP/ESCNB/ESCD).}\label{table:rates_simulations}
\end{footnotesize}
}
\end{table}

\begin{figure}[h!]
\includegraphics[width=\textwidth]{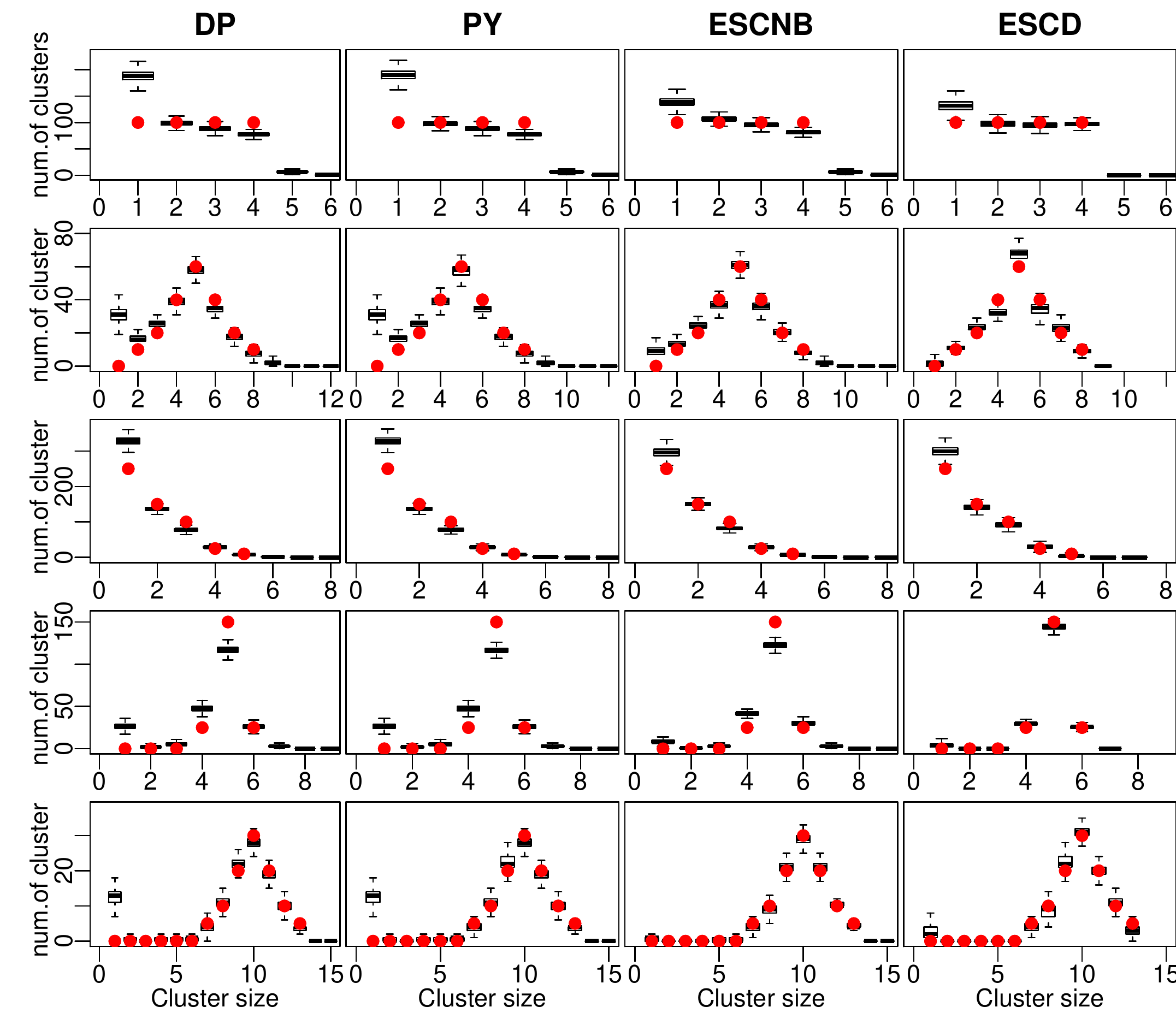}
\caption{Posterior distribution of the number of clusters of each size (black boxplots based on 20k MCMC samples from the posterior after thinning) versus number of clusters of each size in the true data-generating partition (red dots) for $\beta=0.05$. Each column corresponds to a different prior for the partition, and each row to a different data generating partition.} \label{fig:boxplots_simulations}
\end{figure}

As one might expect, Table \ref{table:rates_simulations} illustrates that the FNR and FDR increase with the distortion parameter $\beta$.
For high levels of distortion, such as $\beta=0.1$, all methods struggle to accurately recover the data-generating partition, regardless of the prior specification.
In these context using the ESC-NB or ESC-D prior still helps in reducing the FDR and FNR compared to DP or PYP, although the improvement is less significant. Table \ref{table:rates_simulations} also suggests that partitions with clusters of smaller sizes lead to higher FNR and FDR (e.g.\ scenarios 1 and 3), while those with larger clusters are easier to recover (e.g.\ scenario 5).
This observation supports the idea that microclustering is a hard inferential problem, especially in our motivating applications where cluster sizes are very small. 


\subsection{Hyperparameter settings}\label{sec:settings}
Inference for all the experiments in sections \ref{sec:sim_study} and \ref{sec:experiments} are performed using similar parameter values.
We set $\alpha = 1$ and $\bm{\mu}^{(0)}=\text{NegBin}(r,p)$ for the ESC-D model and choose $\eta_r = s_r = 1$ and $u_p=v_p=2$ for the prior hyperparameters of $r$ and $p$ for both proposed models. We use $p \sim \text{Beta}\left(2, 2\right)$ because a uniform prior implies an unrealistic prior belief that $\mathbb{E}[S_j] = \infty$. For the real data applications, we assume a Beta prior distribution for the distortion probabilities, $\beta_\ell$, with mean 0.005 and standard deviation of 0.01. Finally, to ensure a fair comparison between the two different classes of models, we assign $\text{Gamma}(1, 2/n)$ priors for the concentration parameters of DP and PYP to reflect a vague prior belief that $\mathbb{E}[K] = \frac{n}{2}$. 

\subsection{Real Data Applications}\label{sec:experiments}

In this section, we provide two real data applications that are based upon two panel studies -- the Social Diagnosis Survey (SDS) of quality of life in Poland, and the Survey of Income and Program Participation (SIPP) in the United States. In sections \ref{sec:SDS} and \ref{sec:SIPP} we describe the data sets and present inferential results.

\subsubsection{The Social Diagnosis Survey (SDS)}
\label{sec:SDS}
The SDS is a project that supports diagnosis work derived from institutional indicators of quality of life in households in Poland (anyone older than sixteen years of age). The SDS is based upon panel research, where the first sample was taken in the year 2000 and the same households were revisited approximately every two years thereafter (\url{http://www.diagnoza.com/index-en.html}). 
The SDS database contains 41,227 unique records of individual members of households that participated in the survey in at least one of the years 2011, 2013, and 2015. Thus, individuals are duplicated longitudinally across these three years (waves) but no duplication occurs within a specific year. 
The data is available in horizontal format for the different waves of the survey such that we can uniquely identify duplicated information of the same individual.

\begin{figure}[h!]
\includegraphics[width=\textwidth]{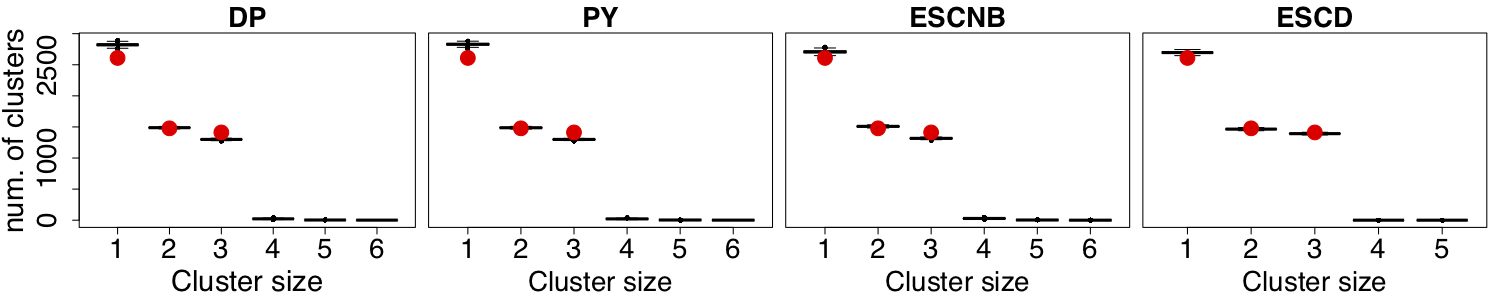}
\caption{\textbf{SDS5500 dataset.} Posterior distribution of the number of clusters of each size (boxplots) versus the number of clusters of each size in the true data partition (red dots) for DP, PY, ESC-NB and ESC-D models for SDS data set of $K=5,500$ unique records.} \label{fig:SDSsizes}
\end{figure}

In order to illustrate Bayesian ER, we construct a subset of $K=5,500$ unique individuals and six fields of information: sex, date of birth (day, month and year), province of residence, and education level -- denoted as SDS5500. Here we consider a subsampled dataset for computational reasons, as performing fully Bayesian ER on the whole SDS database would be computationally challenging (see section \ref{sec:discussion} for discussion). The resulting data set consists of 9,802 records in vertical format with a maximum cluster size of three. 
Note that our aim is to accurately recover the underlying linkage structure for the SDS5500 data set, and that the parameter inferences based on the SDS5500 data set are not meant to reflect the ones based on the entire SDS data set.

Initial exploration of the data showed relatively low levels of noise for most fields within duplicate records making ER for this data straightforward to some extend. Figure \ref{fig:SDSsizes} displays the true number of clusters of each size and the respective posterior distributions. In this case, all the models are able to capture the true distribution of cluster sizes adequately as we expected for such a regular distribution (see also Scenario 3 in section \ref{sec:sim_study}).

\begin{table}[h!]
\centering
\caption{\textbf{SDS5500 dataset.} Posterior mean and standard error of the number of clusters (K), posterior FNR and FDR (in percentages), and posterior estimates of the distortion probabilities of sex, date of birth (year, day, month), province and education level for DP, PY, ESC-NB and ESC-D models for SDS data set of $K=5,500$ unique records.} \label{tab:SDS}
{\renewcommand{\arraystretch}{1.2}
\begin{footnotesize}
\begin{tabular}{c|cccccccccc}
\hline
Model & $\EX [K]$ & SE & FNR & FDR & $\EX [\beta_{1}]$ & $\EX [\beta_{2}]$ & $\EX [\beta_{3}]$ & $\EX [\beta_{4}]$ & $\EX [\beta_{5}]$ & $\EX [\beta_{6}]$\\
\hline 
DP  & 5634 & 15.2 & 5.6 & 3.0 & 
0.0001 & 0.0003  & 0.0071 & 0.0023 & 0.0012 & 0.0937 \\
PY &  5640 & 13.9 & 5.6 & \bftab 2.8 &
0.0009 & 0.0003  & 0.0069 & 0.0019 & 0.0012 & 0.0936\\
ESC-NB  & 5565 & 16.0 & 4.9 & 3.9 & 
0.0012 & 0.0003  & 0.0122 & 0.0039 & 0.0024 & 0.0946\\
ESC-D & \bftab 5553 &  12.0 &  \bftab 4.3 &  3.0 &
0.0011 & 0.0003  & 0.0125 & 0.0044 & 0.0028 & 0.0951\\
\hline
\end{tabular}
\end{footnotesize}
}
\end{table}

As a benchmark for ER, we compare our proposed methodology to the Fellegi-Sunter (FS) method available in the \pkg{RecordLinkage} package in \proglang{R} \citep{FellegiSunter69,SariyarBorg10}.
The FS method is one of the most widely used ER methods in the literature due to its simplicity and speed \citep{FellegiSunter69}. Two records pairs are declared to be matches if a corresponding  likelihood ratio exceeds a threshold.  This particular implementation requires labeled data such that the optimal threshold can be found using an EM algorithm. Thus, by comparing our proposed unsupervised methodology to a semi-supervised one, we are being conservative and favorable towards the FS method.

The resulting FNR and FDR values of the FS method are 4.2\% and 4.4\% with an estimated number of unique entities of $\hat{K}=5,543$.
In agreement with that, Table \ref{tab:SDS} shows that all four models perform relatively well for ER. Overall, the ESC-D model has the best performance with the closest estimate for the true number of unique entities, the lowest FNR (4.3\%) and comparable FDR (3\%) to the DP and PY models. Finally, we observe that the posterior mean for the distortion probability of the education level ($\beta_6$) is considerably higher compared to the other fields in the data. This is consistent with the level of education possibly changing over time between waves of the survey.

\subsubsection{The Survey of Income and Program Participation (SIPP)}
\label{sec:SIPP}

The SIPP is a longitudinal survey that collects information about the income and participation in federal, state, and local programs of individuals and households in the United States \cite{SIPP}.
The SIPP is administered in panels where sample members within each panel are divided into four rotation groups (i.e. subsamples of roughly equal size). One rotation group is interviewed each month such that a wave of the survey consists of a 4-month cycle of interviews.
The data is publicly available through the Inter-university Consortium for Political and Social Research (ICPSR) (\url{https://www.icpsr.umich.edu}). Here, we focus on data from 89,794 unique individuals obtained from the last month of interviews for five waves of the survey performed during 2005 and 2006. Individuals in this subset are only duplicated across waves (not within) and unique identifiers are available. 

\begin{figure}[t!]
\includegraphics[width=\textwidth]{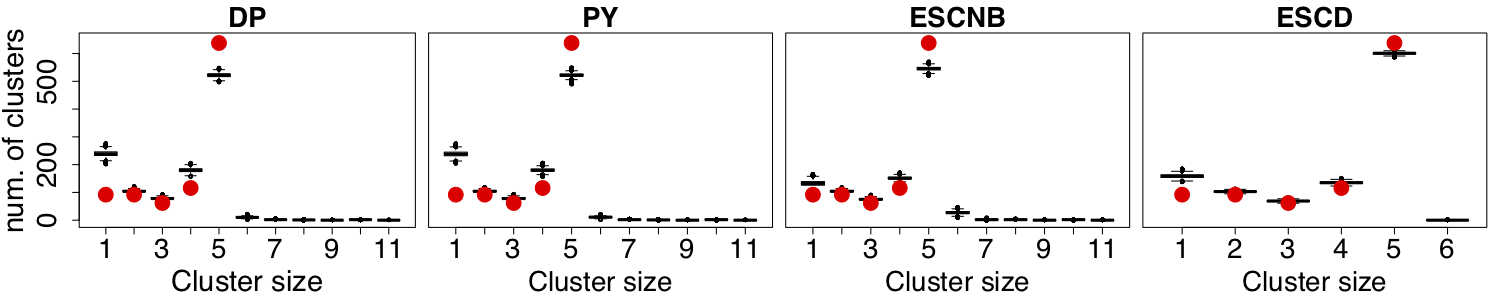}
\caption{\textbf{SIPP1000 dataset.} Posterior distribution of the number of clusters of each size (boxplots) versus the number of clusters of each size in the true data partition (red dots) for DP, PY, ESC-NB and ESC-D models for SIPP data set of $K=1,000$ unique records.} \label{fig:SIPPsizes}
\end{figure}

In this case, we construct a data set of $K=1,000$ unique individuals, which we denote by SIPP1000. SIPP1000 contains a total of 4,116 records and has five fields of information --- sex, year and month of birth, race, and state of residence. The FS method results in a FNR of 10.5\% and FDR of 4.7\% with an estimated number of unique entities of $\hat{K}=1195$. Figure \ref{fig:SIPPsizes} displays the true cluster size distribution of the data which contains 63.8\% clusters of size five and a roughly homogeneous number of clusters of sizes 1 to 4. We observe that the ESC-D model recovers the true distribution of cluster sizes more closely while the DP and PY show inferior performance overall. In particular, the ESC-D model estimates a maximum cluster size of six while the other models identify clusters of sizes up to eleven. Table \ref{tab:SIPP} shows that the ESC-D model displays the lowest error rates (4.8\% and 1.8\%) which represent an error reduction of over 50\% compared to the FS method, and at least 30\% reduction compared to the DP and PY models. Finally, we observe that the month of birth is the field with higher estimated distortion ($\beta_3$) for all models.

\begin{table}[t!]
\centering
\caption{\textbf{SIPP1000 dataset.} Posterior mean and standard error of the number of clusters (K), posterior FNR and FDR (in percentages), and posterior estimates of the distortion probabilities of sex, year and month of birth, race and state for DP, PY, ESC-NB and ESC-D models for SIPP data set of $K=1,000$ unique records.} \label{tab:SIPP}
{\renewcommand{\arraystretch}{1.2}
\begin{footnotesize}
\begin{tabular}{c|ccccccccc}
\hline
Model & $\EX [K]$ & SE & FNR & FDR & $\EX [\beta_{1}]$ & $\EX [\beta_{2}]$ & $\EX [\beta_{3}]$ & $\EX [\beta_{4}]$ & $\EX [\beta_{5}]$ \\
\hline 
DP  & 1139 & 9.4 & 8.0 & 2.6 & 
0.0008 & 0.0004 &  0.0243  & 0.0011 & 0.0005 \\
PY &  1139 & 8.9 & 8.0 & 2.7 &
0.0008 & 0.0004 & 0.0244 & 0.0011 & 0.0005 \\
ESC-NB  & \bftab 1043 & 8.8 & 5.2 & 4.4 & 
0.0009 & 0.0037 &  0.0399  & 0.0020 & 0.0035\\
ESC-D & 1067 &  5.7 &  \bftab 4.8 & \bftab 1.8 &
 0.0006 & 0.0010  & 0.0386 & 0.0028 & 0.0015\\
\hline
\end{tabular}
\end{footnotesize}
}
\end{table}

\section{Discussion} \label{sec:discussion}
In this paper, 
we have proposed a general class of random partition models that satisfy the microclustering property with well-characterized theoretical properties. Our models overcome major limitations in the existing literature on microclustering models.
In our proposed approach, we drop the classical assumption of having an exchangeable sequence of data points, and instead assume an exchangeable sequence of clusters. 
Our proposed framework offers flexibility in terms of the prior distribution of cluster sizes, computational tractability, and applicability to a large number of microclustering tasks (network analysis, genetics, and ER). 
We have established theoretical properties of the resulting class of priors, where we characterize the asymptotic behavior of the number of clusters and of the proportion of clusters of a given size. One appealing feature of our proposed framework is being able to explicitly control the prior distribution of the cluster sizes.
Our framework allows a simple and efficient Markov chain Monte Carlo algorithm to perform statistical inference, where we utilize a partially collapsed Gibbs sampler. 
Finally, the simulated and real experiments showed encouraging results for the performance of our proposed models on the microclustering task of ER.

Our work serves as a first basis for formalizing microclustering in an interpretable and identifiable way with a full characterization of asymptotic properties. We hope that our approach will encourage the emergence of other formalizations of microclustering models in other applications areas. Within ER itself, the most natural extension would be looking at more flexible models for structure data such as textual data. 
In terms of computation, while taking advantage of the Chaperones algorithm allows for computational gains, further exploration of efficient algorithms to improve scalability is needed.
One potential route for our proposed framework would be combining the chaperones algorithm with the locally-balanced schemes proposed by \citet{zanella2019informed}, or exploring crowdsourcing and importance sampling approaches that have been recently proposed by \citep{marchant2017search}. Finally, it would be of interest to develop tighter and more general performance bounds (i.e.\ lower bounds on the misclassification error) as in \cite{steorts2017performance}, which would be appealing for the ER, information theory, and machine learning communities.  
\bibliographystyle{agsm}
\bibliography{references}

\appendix

\newpage




\section{Proofs}
\label{app:proof}
\subsection{Proof of Proposition \ref{prop:marginal_EPPF}}
\label{app:prop1}
\begin{proof}[Proof of Proposition \ref{prop:marginal_EPPF}]
We seek to compute the probability mass function (pmf) of the random partition $\Pi_n=\{C_1,\dots,C_K\}$ obtained from Model \ref{model:esc}. We denote this pmf by $\P(\Pi_n|E_n)$ to make explicit the conditioning on $E_n$ in Step 1 of Model \ref{model:esc}.
Thus,
\begin{equation*}
\P(\Pi_n|E_n)
=
\int
\P(\Pi_n|\bm{\mu},E_n)\P(d\bm{\mu}|E_n)\,.
\end{equation*}
By Bayes' theorem, we find that 
$$
\P(d\bm{\mu}|E_n)
=\frac{P_{\bm{\mu}}(d\bm{\mu})\P(E_n|\bm{\mu})}{\P(E_n)}\,,
$$
where given the construction in Step 1 of Model \ref{model:esc}, we observe that 
$$\P(E_n|\bm{\mu})=\sum_{k=1}^n\sum_{(s_1,\dots,s_k)\in\{1,n\}^k}\mathbb{I}\left(\sum_{j=1}^ks_j=n\right)\prod_{j=1}^k\mu_{s_j}$$
 and 
$\P(E_n)=\int\P(E_n|\bm{\mu})
P_{\bm{\mu}}(d\bm{\mu})$.
Now, consider $\P(\Pi_n|\bm{\mu},E_n)$.
Summing over all possible cluster assignments $\z=(z_1,\dots,z_n)$, we find that
\begin{equation*}
\P(\Pi_n|\bm{\mu},E_n)
=
\sum_{z_1,\dots,z_n=1}^K
\P(\Pi_n|\z,\bm{\mu},E_n)\P(\z|\bm{\mu},E_n).
\end{equation*}
The term $\P(\Pi_n|\z,\bm{\mu},E_n)$ equals 1 for all $K!$ cluster assignments $\z$, leading to the partition $\Pi_n$ and 0 otherwise.
The term $\P(\z|\bm{\mu},E_n)$ equals
\begin{align*}
\P(\z|\bm{\mu},E_n)=&\P(\z|S_1,\dots,S_K)\P(S_1,\dots,S_K|\bm{\mu},E_n)\\
=&\frac{\prod_{j=1}^KS_j!}{n!}\frac{\prod_{j=1}^K \mu_{S_j}}{\P(E_n|\bm{\mu})}\,,
\end{align*}
where $S_j=\sum_{i=1}^n\mathbb{I}(z_i=j)$ denote the size of the $j$-th cluster.
It follows that 
\begin{equation}\label{eq:cond_EPPF_proof}
\P(\Pi_n|\bm{\mu},E_n)
=\frac{K!\prod_{j=1}^KS_j! \mu_{S_j}}{n!\P(E_n|\bm{\mu})}
\end{equation}
and
\begin{equation}\label{eq:marginal_EPPF}
\P(\Pi_n|E_n)
=
\int
\frac{\P(\Pi_n|\bm{\mu},E_n)\P(E_n|\bm{\mu})}{\P(E_n)}P_{\bm{\mu}}(d\bm{\mu})
=
\frac{1}{n!\P(E_n)}\int K!\prod_{j=1}^K |S_j|!\mu_{S_j} P_{\bm{\mu}}(d\bm{\mu})\,.
\end{equation}
The thesis follows from the definition of EPPF.
\end{proof}

\subsection{Proof of Proposition \ref{prop:conditional_eppf} and Corollary \ref{coroll:prediction_rule}}
\label{app:prop2}
\begin{proof}[Proof of Proposition \ref{prop:conditional_eppf} and Corollary \ref{coroll:prediction_rule}]
The expression for the conditional EPPF $p^{(n)}(\cdot;\mu)$ follows directly from Equation \eqref{eq:cond_EPPF_proof}.
The expression for the prediction rule follows from Bayes theorem and 
$$\frac{\P(z_i,\z_{-i}|\bm{\mu},E_n)}{\P(\z_{-i}|\bm{\mu},E_n)}\propto k!\prod_{j=1}^k s_j! \mu_{s_j}\,.$$
\end{proof}

\subsection{Proof of Theorems  \ref{theorem:Kn},  \ref{theorem:fof} and \ref{theorem:microclustering}}
\label{app:thm}
In this section, we prove Theorems  \ref{theorem:Kn}, \ref{theorem:fof} and \ref{theorem:microclustering}. 
The first essential ingredient for our proofs is the Renewal Theorem from the literature on Renewal processes.
\begin{theorem}[Renewal Theorem]\label{thm:renewal_theorem}
Assume $\mu_1>0$ and $\sum_{s=1}s\mu_s\leq \infty$. Then 
$$\P(E_n)\rightarrow \dfrac{1}{\sum_{s=1}s\mu_s} \; \text{as} \; n\to\infty.$$ 
\end{theorem}
We refer to \citet[Thm.2.6]{barbu2009semi} for a proof of the Renewal Theorem. 
The second ingredient is the following technical Lemma that we prove below.
\begin{lemma}\label{lemma:conditioned_as}
Let $X_1,X_2,\dots$ be a sequences of random variables and $E_1,E_2,\dots$ be a sequence of events, with $E_n$ defined on the same probability space of $X_n$.
If $X_n\stackrel{p}\rightarrow c$ as $n\to\infty$ for some $c\in\R$ and $\liminf_{n\to\infty}Pr(E_n)>0$,
then $X_n|E_n\stackrel{p}\rightarrow c$.
\end{lemma}
\begin{proof}
Fix $\epsilon>0$ and define the event
$
A_n=\left\{|X_n-c|>\epsilon \right\}
$.
Since $X_n\stackrel{p}\rightarrow c$ it follows that $\lim_{n\to\infty}Pr(A_n)=0$.
Thus
$$
\limsup_{n\to\infty} Pr(A_n|E_n)
=
\limsup_{n\to\infty}\frac{Pr(A_n\cap E_n)}{Pr(E_n)}
\leq
\frac{\limsup_{n\to\infty}Pr(A_n)}{\liminf_{n\to\infty}Pr(E_n)}
=0\,,
$$
where the last equality follows from $\lim_{n\to\infty}Pr(A_n)=0$ and $\liminf_{n\to\infty}Pr(E_n)>0$.
It follows that, for any $\epsilon>0$,
$\lim_{n\to\infty} Pr(|X_n-c|>\epsilon|E_n)=0$, meaning that $X_n|E_n\stackrel{p}\rightarrow c$.
\end{proof}

\begin{proof}[Proof of Theorem \ref{theorem:Kn}]
We use $\mathcal{L}(\cdot)$ and $\mathcal{L}(\cdot|\cdot)$ to denote marginal and conditional distributions of random variables. 
By construction of $\Pi_n\sim ESC_{[n]}(\bm{\mu})$, we have 
\begin{align}
\mathcal{L}(K_n)=\mathcal{L}(Y_n|E_n)\;\hbox{ and }\;\mathcal{L}(S_j)=\mathcal{L}(X_j|E_n)&& n\geq 1;j=1,\dots,K_n\label{eq:K_n_cond}
\end{align}
where $X_1,X_2,\dots\stackrel{iid}\sim\bm{\mu}$, $Y_n=\max\{k\,:\,\sum_{j=1}^k X_j\leq  n\}$
and 
\begin{equation}\label{eq:E_n_appendix}
E_n=\left\{\omega\in\Omega\,:\, \hbox{ for some }k\geq 1\hbox{ it holds } \sum_{j=1}^k X_j =n
\right\}\,.
\end{equation}
Theorem \ref{thm:renewal_theorem} implies $\liminf_{n\to\infty}Pr(E_n)>0$.
Also, the strong law of large numbers for renewal processes (see e.g.\ 
\citealp[Thm.2.3]{barbu2009semi}) 
implies that 
$n^{-1}Y_n$ converges almost surely to $ (\sum_{s=1}^\infty s\mu_s)^{-1}$, and thus, also in probability. 
Since $n^{-1}Y_n\stackrel{p}\to (\sum_{s=1}^\infty s\mu_s)^{-1}$ and $\liminf_{n\to\infty}Pr(E_n)>0$, it follows by Lemma \ref{lemma:conditioned_as} and Equation \eqref{eq:K_n_cond} that $n^{-1}K_n\stackrel{p}\to (\sum_{s=1}^\infty s\mu_s)^{-1}$, as desired.
\end{proof}

\begin{proof}[Proof of Theorem \ref{theorem:fof}]
By construction of $\Pi_n\sim ESC_{[n]}(\bm{\mu})$ we have 
\begin{equation}\label{eq:M_n_conditioned}
\mathcal{L}(M_{s,n})=\mathcal{L}(L_{s,n}|E_n)\qquad n\geq 1\,,
\end{equation}
where 
$L_{s,n}=\sum_{j=1}^{Y_n}\1(X_j=s)$, and $X_j$, $Y_n$ and  $E_n$ are defined as in the proof of Theorem \ref{theorem:Kn}.
Since $\1(X_j=s)$ are independent and identically distributed Bernoulli random variables with mean $\mu_s$ and $\lim_{n\to\infty}Y_n=\infty$ almost surely, the strong law of large numbers imply that
\begin{equation}\label{eq:L_over_Y_conv}
\lim_{n\to\infty}\frac{L_{s,n}}{Y_n}
=
\lim_{n\to\infty}\frac{\sum_{j=1}^{Y_n}\1(X_j=s)}{Y_n}
=
\lim_{n\to\infty}\frac{\sum_{j=1}^{n}\1(X_j=s)}{n}
=\mu_s
\qquad \hbox{almost surely}
\,.
\end{equation}
Thus,
$$
\lim_{n\to\infty}\frac{L_{s,n}}{n}
=
\lim_{n\to\infty}\frac{L_{s,n}}{Y_n}\frac{Y_n}{n}
=\mu_s\left(\sum_{\ell=1}^{\infty}\ell\mu_\ell\right)^{-1}
\qquad \hbox{almost surely}
\,,
$$
where we used the fact that $\lim_{n\to\infty}n^{-1}Y_n=\left(\sum_{\ell=1}^{\infty}\ell\mu_\ell\right)^{-1}$ almost surely by the strong law of large numbers for renewal processes (see e.g.\ 
\citealp[Thm.2.3]{barbu2009semi}
Since almost sure convergence implies convergence in probability, we have $n^{-1}L_{s,n}\stackrel{p}\to \mu_s\left(\sum_{\ell=1}^{\infty}\ell\mu_\ell\right)^{-1}$, which implies $n^{-1}M_{s,n}\stackrel{p}\to \mu_s\left(\sum_{\ell=1}^{\infty}\ell\mu_\ell\right)^{-1}$ by Equation \eqref{eq:M_n_conditioned} and Lemma \ref{lemma:conditioned_as}, as desired.

Consider now part (b). 
The size of cluster chosen uniformly at random from the clusters of $\Pi_n$ is a random variable
$S_{U_n}$, where $S_1,\dots, S_{K_n}$ are the sizes of the clusters of $\Pi_n$ and $U_n$ is a random variable satisfying $U_n|\Pi_n\sim \text{Uniform}\{1,\dots,K_n\}$.
For any positive integer $s$, by the definition of $U_n$, we have $Pr(S_{U_n}=s|\Pi_n)=K_n^{-1}M_{s,n}$ and thus
\begin{equation}\label{eq:prob_S_U}
Pr(S_{U_n}=s)=\E[Pr(S_{U_n}=s|\Pi_n)]=\E\left[\frac{M_{s,n}}{K_n}\right]\,.
\end{equation}
By construction of $\Pi_n\sim ESC_{[n]}(\bm{\mu})$, we have 
$$
\mathcal{L}\left(\frac{M_{s,n}}{K_n}\right)=\mathcal{L}\left(\left.\frac{L_{s,n}}{Y_n}\;\right|E_n\right)\qquad n\geq 1\,,
$$
and by Equation \eqref{eq:L_over_Y_conv} we have $Y_n^{-1} L_{s,n}\stackrel{p}\to \mu_s$.
Thus Lemma \ref{lemma:conditioned_as} implies $K_n^{-1}M_{s,n}\stackrel{p}\to \mu_s$.
Since $K_n^{-1}M_{s,n} \in[0,1]$ it follows that $\E\left[K_n^{-1}M_{s,n}\right]\to \mu_s$ and thus, by Equation \eqref{eq:prob_S_U}, $Pr(S_{U_n}=s)\to \mu_s$ as desired.
\end{proof}

\begin{proof}[Proof of Theorem \ref{theorem:microclustering}]
Let $X_j$, $Y_n$ and  $E_n$ be defined as in the proof of Theorem \ref{theorem:Kn}.
By construction of $\Pi_n\sim ESC_{[n]}(\bm{\mu})$,
 we have 
\begin{equation}\label{eq:M_n_cond}
\mathcal{L}(M_{n})=\mathcal{L}(L_n|E_n)\qquad n\geq 1
\end{equation}
where $L_n=\max\{X_1,\dots,X_{Y_n}\}$.
For any $\epsilon>0$ consider 
\begin{align*}
\Pr\left(n^{-1}L_n>\epsilon\right)
&=
\Pr\left(n^{-1}\max\{X_1,\dots,X_{Y_n}\}>\epsilon\right)
\leq
\Pr\left(n^{-1}\max\{X_1,\dots,X_{n}\}>\epsilon\right)
\\
&=1-\Pr\left(\cap_{j=1}^n\{X_j\leq n\epsilon\}\right)=1-\Bigg(\sum_{j=1}^{\lceil\epsilon n\rceil}\mu_j\Bigg)^n\,,
\end{align*}
where the inequality in the first row of the display follows from $Y_n\geq n$.
Since 
$1-x^n\leq n(1-x)$
for all $x\in[0,1]$ and $n\geq 1$, we have
\begin{align*}
1-
\Bigg(\sum_{j=1}^{\lceil\epsilon n\rceil}\mu_j\Bigg)^n
&\;\leq\;
n\Bigg(1-\sum_{j=1}^{\lceil\epsilon n\rceil}\mu_j\Bigg)
\;=\;
n\sum_{j=\lfloor\epsilon n\rfloor+1}^\infty\mu_j \\
&=
\epsilon^{-1}\sum_{j=\lfloor\epsilon n\rfloor+1}^\infty\epsilon n\mu_j
\;\leq\;
\epsilon^{-1}\sum_{j=\lfloor\epsilon n\rfloor+1}^\infty j\mu_j
\rightarrow 0
\quad\hbox{as }n\to\infty\,,
\end{align*}
where the convergence $\lim_{n\to\infty}\sum_{j=\lfloor\epsilon n\rfloor+1}^\infty j\mu_j
= 0$  follows from  
$\sum_{j=1}^\infty j \mu_j<\infty$ and $\lim_{n\to\infty}\lfloor\epsilon n\rfloor+1=\infty$.
Combining the last inequalities we obtain 
$
\lim_{n\to\infty}\Pr\left(n^{-1} L_n>\epsilon\right)
\rightarrow 0
$ or, in other words,
$n^{-1} L_n\stackrel{p}\to 0$ as $n\rightarrow\infty $.
{Thus, by Equation \eqref{eq:M_n_cond}, Lemma \ref{lemma:conditioned_as}, and $\liminf_{n\to\infty}Pr(E_n)>0$ (which follows from Theorem \ref{thm:renewal_theorem}), we obtain $n^{-1} M_n \stackrel{p}\rightarrow 0$ as $n\rightarrow\infty $, as desired.}
\end{proof}

\section{{Samplers for Posterior Inference}} \label{supp:samplers}
{In this section, we provide additional details regarding the samplers used for posterior inference.}
In the following derivations, we use the fact that, under the $ESC_{[n]}(P_{\bm{\mu}})$ model, the joint distribution of $\bm{\mu}$ and $\Pi_n$ is
\begin{equation}\label{eq:joint_mu_Pi}
\P(d\bm{\mu},\Pi_n)
=
\frac{P_{\bm{\mu}}(d\bm{\mu})}{\P(E_n)}
\frac{K!}{n!}\prod_{j=1}^K S_j! \mu_{S_j}\,,
\end{equation}
{which can be easily derived using Equation \eqref{eq:eppf_marginal}.}
It follows that the conditional distribution of $\bm{\mu}$ given $\Pi_n$ satisfies
\begin{equation}\label{eq:mu_given_Pi}
\P(d\bm{\mu}|\Pi_n)
\propto P_{\bm{\mu}}(d\bm{\mu})  \prod_{j=1}^K S_j! \mu_{S_j}\,.
\end{equation}
The precise mathematical interpretation of Equation \eqref{eq:mu_given_Pi} is that the Radon–Nikodym derivative between the distribution of $\bm{\mu}$ conditional on $\Pi_n$ and the distribution $P_{\bm{\mu}}$ is proportional to $\prod_{j=1}^K S_j! \mu_{S_j}$.
The key aspect of Equation \eqref{eq:mu_given_Pi} is that the conditional distribution of $\bm{\mu}$ does not depend on the intractable term $\P(E_n|\bm{\mu})$, which makes the updates of $\bm{\mu}|\Pi_n$ in the MCMC algorithms for posterior sampling straightforward.


\subsection{ESC-NB model}
Recall that, for the ESC-NB model, $\bm{\mu}=\bm{\mu}(r,p)$ is a deterministic function of $r$ and $p$ specified by Equation  \eqref{eq:mu_esc_NB}.
\begin{proof}[Derivation of Equation \eqref{eq:cond_par_ESC_NB}]
Since $r,p$ are conditionally independent of $\x$ given $\Pi_n$ we have $\Pr(r,p|\Pi_n,\x)=\Pr(r,p|\Pi_n) $.
Then, combining Equation \eqref{eq:mu_given_Pi} with Equation \eqref{eq:mu_esc_NB} and the prior specification $r\sim Gamma(\eta_r,s_r)$, $p\sim Beta(u_p,v_p)$, we obtain
\begin{align*}
\Pr(r,p|\Pi_n,\x) &=\Pr(r,p|\Pi_n) \propto 
\left(\frac{r^{\eta_r-1}e^{-\frac{r}{s_r}}}{\Gamma(\eta_r)s_r^{\eta_r}}\right)
\left(\frac{p^{ u_p-1}(1-p)^{ v_p-1}}{B( u_p, v_p)}\right)
\prod_{j=1}^K S_j! \mu_{S_j}\nonumber\\
&\propto
r^{\eta_r-1}e^{-\frac{r}{s_r}}
p^{u_p-1}(1-p)^{ v_p-1}
\prod_{j=1}^K S_j! \gamma\frac{\Gamma(S_j+r)p^{S_j} }{\Gamma(r)S_j!}
\\&
\propto
r^{\eta_r-1}e^{-\frac{r}{s_r}}
p^{ n+u_p-1}(1-p)^{ v_p-1}
\gamma^K\prod_{j=1}^K\frac{\Gamma(S_j+r)}{\Gamma(r)}\,,
\end{align*}
which proves Equation \eqref{eq:cond_par_ESC_NB}.
\end{proof}

\begin{proof}[Derivation of Equation \eqref{eq:full_cond_ESCNB}]
Given the dependence structure of $(r,p)$, $\Pi_n$, and $\x$, we have $\Pr(\Pi_n|r,p,\x)\propto\Pr(\Pi_n|r,p)\Pr(\x|\Pi_n)$ and thus
\begin{equation}
\P(z_i=j|\z_{-i},\x,r,p)
\propto
\Pr(\x|\z_{-i},z_i=j)\times
\P(z_i=j|\z_{-i},r,p)\,.
\end{equation}
Corollary \ref{coroll:prediction_rule} implies
\begin{align*}
\P(z_i=j|\z_{-i},r,p)
&\propto
\left\{
\begin{array}{ll}
(S_j+1)
\frac{\mu_{(S_j+1)}}{\mu_{S_j}}& \hbox{if }j=1,\dots,K_{-i},
\\
(K_{-i}+1)\mu_1 & \hbox{if }j=K_{-i}+1\,,
\end{array}
\right.
\end{align*}
where, by  Equation \eqref{eq:mu_esc_NB}, we have $\mu_1
=\gamma\,r\,p$ and 
\begin{align*}
\frac{\mu_{(S_j+1)}}{\mu_{S_j}}
=
\frac{\gamma\frac{\Gamma(S_j+1+r)p^{S_j+1} }{\Gamma(r)(S_j+1)!}}{\gamma\frac{\Gamma(S_j+r)p^{S_j} }{\Gamma(r)S_j!}}
=
p
\frac{\Gamma(S_j+1+r) }{\Gamma(S_j+r)}
\frac{S_j!}{(S_j+1)!}
=
p\frac{S_j+r}{S_j+1}.
\end{align*}
Therefore
\begin{align*}
\P(z_i=j|\z_{-i},r,p)
&\propto
\left\{
\begin{array}{ll}
(S_j+1)p
\frac{S_j+r}{S_j+1}
& \hbox{if }j=1,\dots,k_{-i},
\\
(K_{-i}+1)\gamma\,r\,p & \hbox{if }j=k_{-i}+1\,,
\end{array}
\right.
\\
&\propto
\left\{
\begin{array}{ll}
S_j+r & \hbox{if }j=1,\dots,K_{-i},
\\
(K_{-i}+1)\gamma r & \hbox{if }j=K_{-i}+1\,.
\end{array}
\right.
\end{align*}
\end{proof}

\subsection{ESC-D model}
\begin{proof}[Derivation of Equation \eqref{eq:ESC_D_parameters}]
While in the ESC-NB model $\bm{\mu}$ is a deterministic function of $r$ and $p$, for the ESC-D model we have $\bm{\mu}|r,p\sim Dir(\alpha,\bm{\mu}^{(0)})$, where $\bm{\mu}^{(0)}=\bm{\mu}^{(0)}(r,p)$ is defined in Equation \eqref{eq:mu0_esc}.
Thus, integrating out $\bm{\mu}$ in Equation \eqref{eq:joint_mu_Pi} and using $r\sim Gamma(\eta_r,s_r)$ and $p\sim Beta(u_p,v_p)$, we obtain 
\begin{align}
P(r,p,\Pi_n) &=
\frac{1}{P(E_n)}\left(\frac{r^{\eta_r-1}e^{-\frac{r}{s_r}}}{\Gamma(\eta_r)s_r^{\eta_r}}\right)
\left(\frac{p^{ u_p-1}(1-p)^{ v_p-1}}{B( u_p, v_p)}\right)
\E_{\bm{\mu}\sim Dir(\alpha,\bm{\mu}^{(0)})}\left[
\frac{K!}{n!}\prod_{j=1}^K S_j! \mu_{S_j}
\right]\,.\label{eq:joint_ESC_D}
\end{align}
Using $M_{s,n}=\sum_{j=1}^K\1(S_j=s)$ and standard expressions for the moments of the Dirichlet distribution we obtain
\begin{align}
\E_{\bm{\mu}\sim Dir(\alpha,\bm{\mu}^{(0)})}\left[
\frac{K!}{n!}\prod_{j=1}^K S_j! \mu_{S_j}
\right]&=\frac{K!}{n!}
\left(\prod_{s=1}^{M_n} s!^{M_{s,n}}\right)
\E_{\bm{\mu}\sim Dir(\alpha,\bm{\mu}^{(0)})}\left[
\prod_{s=1}^{M_n} \mu_{s}^{M_{s,n}}
\right]
\nonumber\\
&=
\frac{K!}{n!}
\left(\prod_{s=1}^{M_n} s!^{M_{s,n}}\right)
\frac{\Gamma(\alpha)}{\Gamma(K+\alpha)}
\prod_{s=1}^{M_n} \frac{\Gamma(M_{s,n}+\alpha\,\mu^{(0)}_s)}{\Gamma(\alpha\,\mu^{(0)}_s)}
\nonumber\\
&=\frac{K!}{n!}
\frac{\Gamma(\alpha)}{\Gamma(K+\alpha)}
\prod_{s=1}^{M_n} \frac{s!^{M_{s,n}}\Gamma(M_{s,n}+\alpha\,\mu^{(0)}_s)}{\Gamma(\alpha\,\mu^{(0)}_s)}\,.
\label{eq:Dir_moments}
\end{align}
Combining Equations \eqref{eq:joint_ESC_D} and \eqref{eq:Dir_moments} we obtain that the joint distribution of $r$, $p$ and $\Pi_n$ under the ESC-D model satisfies
\begin{align}
P(r,p,\Pi_n) &\propto
\frac{r^{\eta_r-1}e^{-\frac{r}{s_r}}p^{ u_p-1}(1-p)^{ v_p-1}K!}{\Gamma(K+\alpha)}
\prod_{s=1}^{M_n} \frac{s!^{M_{s,n}}\Gamma(M_{s,n}+\alpha\,\mu^{(0)}_s)}{\Gamma(\alpha\,\mu^{(0)}_s)}
\,.\label{eq:joint_ESC_D_prop}
\end{align}
The expression in Equation \eqref{eq:ESC_D_parameters} follows from Equation \eqref{eq:joint_ESC_D_prop} and the fact that $\Pr(r,p|\Pi_n,\x)=\Pr(r,p|\Pi_n) $ because $r$ and $p$ are conditionally independent of $\x$ given $\Pi_n$.
\end{proof}

\section{Importance Sampler for $ESC$ models}\label{sec:IS}

In this section we describe an importance sampler that can be used to generate weighted samples from random partitions $\Pi_n\sim ESC_{[n]}(P_{\bm{\mu}})$.
The propose algorithm is not a fully standard importance sampler and thus we prove its validity in Theorem \ref{eq:thm_IS}.
In the context of Bayesian inferences, this algorithm can be used to generate samples from a $ESC_{[n]}(P_{\bm{\mu}})$ prior distribution for random partition.
Unlike the rejection sampler described in the main document, we expect the importance sampler described here to be efficient even when $\E_{\bm{\mu}\sim P_{\bm{\mu}}}\left[(\sum_{s=1}s\mu_s)^{-1}\right]$ becomes small.
\begin{alg}\label{alg:IS}\emph{(Importance Sampler for ESC models)}
\begin{enumerate}
\item Sample $\bm{\mu}\sim P_{\bm{\mu}}$ and $S_1,\dots,S_{R}|\bm{\mu}\stackrel{iid}\sim\bm{\mu}$ until the first value $R$ such that $\sum_{j=1}^{R}S_j\geq n$.
\item For $k=1,\dots,R$ define $D_k=n-\sum_{j=1}^{k-1}S_j$ and 
 $W=\sum_{k=1}^{R}\mu_{D_k}$.
\item Sample $K$ from $\{1,\dots,R\}$ with probability $\P(K=k)=\mu_{D_k}/W$, and 
define the cluster allocation variables $(z_1,\dots,z_n)$ as a uniformly at random permutation of the vector 
\begin{equation}\label{eq:indices_order_2}
(\underbrace{1,\ldots,1}_\text{$S_1$
  times},\underbrace{2,\ldots,2}_\text{$S_2$
  times},\ldots\ldots,\underbrace{K-1,\ldots,K-1}_\text{$S_{K-1}$ times},\underbrace{K,\ldots,K}_\text{$D_K$ times}).
\end{equation}
\item Output the resulting partition  $\Pi_n$ as a weighted sample from the model $ESC_{[n]}(P_{\bm{\mu}})$ with importance weight $\P(E_n)^{-1}W$.
\end{enumerate}
\end{alg}
Intuitively, given each vector of cluster sizes  $(S_1,\dots,S_{k-1})$, Algorithm \ref{alg:IS} considers the probability $\mu_{D_k}$ of sampling $S_k=D_k$ 
 and weights the resulting vector of cluster sizes $(S_1,\dots,S_{k-1},D_k)$ accordingly.
The following theorem shows that the algorithm is valid, in the sense that it returns weighted samples from the distribution $ESC_{[n]}(P_{\bm{\mu}})$ that produce unbiased and consistent Monte Carlo estimators like standard Importance Sampling does. 
\begin{theorem}\label{eq:thm_IS}
For every real-valued function $h$ defined over the space of partitions of $[n]$ we have 
\begin{equation}\label{eq:IS_unbiased}
\E_{(\Pi_n,W)\sim Alg\ref{alg:IS}}[\P(E_n)^{-1}W\,h(\Pi_n)]
=
\E_{\Pi_n\sim ESC_{[n]}(P_{\bm{\mu}})}[h(\Pi_n)]\,,
\end{equation}
where the notation $(\Pi_n,W)\sim Alg\ref{alg:IS}$ means that $\Pi_n$ is a partition produced by Algorithm \ref{alg:IS} with associated importance weight $\P(E_n)^{-1}W$.
Also, given a sequence $(\Pi^{(t)}_n,W^{(t)})_{t=1}^\infty\stackrel{iid}\sim Alg\ref{alg:IS}$, 
we have
\begin{equation}\label{eq:self_norm_IS}
\frac{\sum_{t=1}^TW^{(t)}\,h(\Pi_n^{(t)})}{\sum_{t=1}^TW^{(t)}}
\stackrel{a.s.}\to
\E_{\Pi_n\sim ESC_{[n]}(P_{\bm{\mu}})}[h(\Pi_n)]
\qquad\hbox{as }T\to\infty\,.
\end{equation}
\end{theorem}
\begin{proof}
By Proposition \ref{prop:marginal_EPPF}, or equivalently by \eqref{eq:marginal_EPPF}, we have
\begin{align}\label{eq:proof_IS_eq1}
\E_{\Pi_n\sim ESC_{[n]}(P_{\bm{\mu}})}[h(\Pi_n)]=
\frac{1}{n!\P(E_n)}
\sum_{\Pi_n}h(\Pi_n)
\int K!\prod_{j=1}^K |S_j|!\mu_{S_j} P_{\bm{\mu}}(d\bm{\mu})\,,
\end{align}
where the sum over $\Pi_n$ runs over all partitions of $[n]$.
We now consider the expectation $\E_{(\Pi_n,W)\sim Alg\ref{alg:IS}}[\P(E_n)^{-1}W\,h(\Pi_n)]$ and show that it is equal to the same expression.
To simplify the proof we consider an equivalent formulation of Algorithm \ref{alg:IS}, where we simulate $\bm{\mu}\sim P_{\bm{\mu}}$ and $S_1,\dots,S_{n}|\bm{\mu}\stackrel{iid}\sim\bm{\mu}$ in Step 1; we set $W=\sum_{k=1}^n\mu_{D_k}$ with $\mu_{D_k}=0$ when $D_k\leq 0$ in Step 2; we sample $K$ from $\{1,\dots,n\}$ with probability $\P(K=k)=\mu_{D_k}/W$ in Step 3 and leave the rest of the algorithm unchanged. 
The latter is an equivalent formulation of Algorithm \ref{alg:IS} that is computationally less efficient because it generates additional variables $S_{R+1},\dots,S_n$ that are not necessary in practice, but is slightly simpler to analyse because it avoids the use of the auxiliary variable $R$.
In order to keep the notation light, we denote $\S=(S_1,\dots,S_n)$ and $\z=(z_1,\dots,z_n)$ and we denote random variables (e.g.\ $\S$, $K$ and $\z$) and their possible realizations with the same symbols.
We have
\begin{align}
&\E_{(\Pi_n,W)\sim Alg\ref{alg:IS}}[W\,h(\Pi_n)]\nonumber
\\&=
\int
\sum_{\S\in\{1,2,\dots\}^n}
\P(\S|\bm{\mu})
\sum_{K=1}^n
\P(K|\S,\bm{\mu})
\sum_{\z}
\P(z|K,\S) W h(\Pi_n(\z)) P_{\bm{\mu}}(d\bm{\mu})\nonumber
\\&=
\int\sum_{\S\in\{1,2,\dots\}^n}
\left(\prod_{j=1}^n\mu_{S_j}\right)
\sum_{K=1}^n
\frac{\mu_{D_K}}{\sum_{k=1}^n \mu_{D_k}}
\sum_{\z}
\frac{\left(\prod_{j=1}^{K-1} S_j!\right)D_{K}!}{n!}
\left(\sum_{k=1}^n \mu_{D_k}\right) h(\Pi_n(\z)) P_{\bm{\mu}}(d\bm{\mu})\nonumber
\\&=
\frac{1}{n!}
\int
\sum_{\S\in\{1,2,\dots\}^n}
\sum_{K=1}^n
\left(\prod_{j=1}^n\mu_{S_j}\right)
\mu_{D_K}\left(\prod_{j=1}^{K-1} S_j!\right)D_{K}!
\sum_{\z}
h(\Pi_n(\z))P_{\bm{\mu}}(d\bm{\mu})\,,\label{eq:proof_IS_eq3}
\end{align}
where the sum over $\z$ runs over all the vectors that can be obtained as a permutation of the vector in \eqref{eq:indices_order_2}.
Reorganizing the sum and exploiting the fact that $\z$ and $\Pi_n$ depend only on $(S_1,\dots,S_{K-1})$ and $K$, we can integrate out $(S_{K},\dots,S_{n})$ and write \eqref{eq:proof_IS_eq3} as
\begin{align*}
\frac{1}{n!}
\sum_{K=1}^n
\sum_{(S_1,\dots,S_{K-1})\in\{1,2,\dots\}^{K-1}}
\sum_{\z}
h(\Pi_n(\z))
\int
\left(\prod_{j=1}^{K-1}\mu_{S_j}S_j!\right)
\mu_{D_K}D_K!
P_{\bm{\mu}}(d\bm{\mu})\,.
\end{align*}
Re-writing the sums above in terms of the resulting partition $\Pi_n$, and exploiting the fact that each partition $\Pi_n$ can be obtained through $K!$ different cluster assignments $\z$, we have
\begin{align}\label{eq:proof_IS_eq2}
&\E_{(\Pi_n,W)\sim Alg\ref{alg:IS}}[W\,h(\Pi_n)]
=
\frac{1}{n!}
\sum_{\Pi_n}h(\Pi_n)
K!\left(\prod_{j=1}^{K-1} |S_j|!\mu_{S_j}\right)\mu_{D_K}D_K!\,,
\end{align}
where the sum over $\Pi_n$ runs over all partitions of $[n]$ and the cluster sizes of $\Pi_n$ are denoted as $(S_1,\dots,S_{K-1},D_K)$ for coherence with the notation of Algorithm \ref{alg:IS}.
Comparing \eqref{eq:proof_IS_eq1} and \eqref{eq:proof_IS_eq2} we obtain \eqref{eq:IS_unbiased}.

The almost sure convergence in \eqref{eq:self_norm_IS} follows by applying the strong law of large numbers to both numerator and denominator in the fraction on the left-hand side, and then noting that by \eqref{eq:IS_unbiased} we have 
$$
\frac{\E_{(\Pi_n,W)\sim Alg\ref{alg:IS}}[W\,h(\Pi_n)]}{\E_{(\Pi_n,W)\sim Alg\ref{alg:IS}}[W]}
=
\frac{\P(E_n)
\E_{\Pi_n\sim ESC_{[n]}(P_{\bm{\mu}})}[h(\Pi_n)]}{\P(E_n)}
=
\E_{\Pi_n\sim ESC_{[n]}(P_{\bm{\mu}})}[h(\Pi_n)]
\,.
$$
\end{proof}

Note that the normalized importance weight $\P(E_n)^{-1}W$ involves the constant $\P(E_n)$ that is typically not available in closed form.
However, this is not a problem because the self-normalized importance sampling estimator defined in \eqref{eq:self_norm_IS} is not sensitive to multiplicative constants in the importance weights. Thus, one can directly use $W$ as an importance weight, ignoring the unknown constant $\P(E_n)$.


\section{Likelihood Derivation for Entity Resolution} \label{supp:likelihood}
{In this section, we provide the derivation of the likelihood that is used in our ER task.}
Recall that the observed data $\x$ consist of $n$ records $(x_i)_{i=1}^n$ and each record $x_i$ contains $L$ fields $(x_{i\ell})_{\ell=1}^L$. Each field $\ell$ is associated to two hyperparameters: a distortion probability $\beta_\ell\in(0,1)$ and a density vector
$\btheta_{\ell}=(\theta_{\ell d})_{d=1}^{D_\ell}\in[0,1]^{D_\ell}$, where $D_\ell$ denotes the number of categories for field $\ell$ and $\sum_{d=1}^{D_\ell}\theta_{\ell d}=1$.
As mentioned in Section \ref{sec:RL}, we assume that clusters are conditionally independent given the partition $\Pi_n$ and the hyperparameters $\bbeta=(\beta_\ell)_{\ell=1}^L$ and $\btheta=(\btheta_{\ell})_{\ell=1}^L$, resulting in 

\begin{equation}\label{eq:independent_fields_assumption}
P(\x|\Pi_n,\bbeta,\btheta)
=
\prod_{j=1}^K
\prod_{\ell=1}^L
P(\x_{j\ell}|\beta_\ell,\btheta_\ell)\,,
\end{equation}
where $\x_{j\ell}=\{x_{i\ell}\,:\,i\in C_j\}$.
For each $C_j \in \Pi_n$, the distribution of $\x_{j\ell}|\beta_\ell,\btheta_\ell$ is given by
\begin{align}
y_{j\ell}&\sim \btheta_\ell\label{eq:y_gen1}\\
x_{i\ell}|y_{j\ell}&\stackrel{iid}\sim \beta_\ell\btheta_\ell + (1-\beta_\ell)\delta_{y_{j\ell}}\qquad i\in C_j
\,,\label{eq:x_giv_y1}
\end{align}
where $y_{j\ell}$ represent the correct $l$-th feature of the entity associated to cluster $C_j$, and $\beta_\ell$ is the probability of distortion in feature $\ell$.
Integrating out $y_{j\ell}$ from Equations \eqref{eq:y_gen1} and \eqref{eq:x_giv_y1} it follows
\begin{align}
P(\x_{j\ell}|\beta_\ell,\btheta_\ell)
&=
\sum_{d=1}^{D_\ell}
P(y_{j\ell}=d|\btheta_\ell)
\prod_{i\in C_j}P(x_{i\ell}|\beta_\ell,y_{j\ell}=d)\nonumber\\
&=
\sum_{d=1}^{D_\ell}
\theta_{\ell d}
\prod_{i\in C_j}(\beta_\ell\theta_{\ell x_{n\ell}}+(1-\beta_\ell)\mathbbm{1}(x_{i\ell}=d))
\nonumber\\
&=
\left(\prod_{i\in C_j}\beta_\ell\theta_{\ell x_{i\ell}}\right)
\sum_{d=1}^{D_\ell}
\theta_{\ell d}
\prod_{i\in C_j}
\frac{(\beta_\ell\theta_{\ell x_{i\ell}}+(1-\beta_\ell)\mathbbm{1}(x_{i\ell}=d))}{\beta_\ell\theta_{\ell x_{i\ell}}}
\label{eq:integrating_y_first}
\end{align}
To proceed we denote by $x^{(j)}_{1\ell},\dots,x^{(j)}_{m^{(j)}\ell}$ the collection of unique values in $\x_{j\ell}$ and by $q^{(j)}_{1\ell},\dots,q^{(j)}_{m^{(j)}\ell}$ the corresponding frequencies, meaning that for each $i\in\{1,\dots,m^{(j)}\}$ the value $x^{(j)}_{i\ell}$ appears exactly $q^{(j)}_{i\ell}$ times in $\x_{j\ell}$.
Then from Equation \eqref{eq:integrating_y_first} we have 
\begin{align}\label{eq:integrating_y_second}
P(\x_{j\ell}|\beta_\ell,\btheta_\ell)
&=
\left(\prod_{i\in C_j}\beta_\ell\theta_{\ell x_{i\ell}}\right)
f(\x_{j\ell},\beta_\ell,\btheta_\ell)\,,
\end{align}
where
\begin{align}\label{eq:f}
f(\x_{j\ell},\beta_\ell,\btheta_\ell)
&=
1-\sum_{i=1}^{m^{(j)}}\theta_{\ell x^{(j)}_{i\ell}}
+\sum_{i=1}^{m^{(j)}}\theta_{\ell x^{(j)}_{i\ell}}
\left(
\frac{\beta_\ell\theta_{\ell x^{(j)}_{i\ell}}+(1-\beta_\ell)}{\beta_\ell\theta_{\ell x^{(j)}_{i\ell}}}
\right)^{q^{(j)}_{i\ell}}
\,.
\end{align}
Combining Equations \eqref{eq:integrating_y_second} and \eqref{eq:independent_fields_assumption}, we obtain the desired likelihood function
\begin{equation}\label{eq:likelihood_function}
P(\x|\Pi_n,\bbeta,\btheta)
=
\left(
\prod_{\ell=1}^L
\prod_{i=1}^n
\beta_\ell
\theta_{\ell x_{i\ell}}
\right)
\prod_{j=1}^K
\prod_{\ell=1}^L
f(\x_{j\ell},\beta_\ell,\btheta_\ell)
\,.
\end{equation}

\clearpage

\section{Implementation details of MCMC algorithms}\label{supp:convergence}

In this section, we provide more details on the MCMC algorithms used to approximate posterior quantities of interest in Section \ref{sec:sim_both} of the paper. 
Posterior computation is performed using the samplers described in Sections \ref{sec:ESCNB}-\ref{sec:ESCD} of the paper.
The results are based on MCMC runs of $2\times 10^7$ iterations, thinning every $1,000$ iterations\footnote{More precisely, we perform $2\times 10^4$ MCMC iterations, and within each iteration perform one update of the global parameters and 1000 updates of the partition given the global parameters using the chaperones algorithm.} and then discarding the first 5000 out of 20000 resulting samples as burn-in.
In all cases standard convergence diagnostics and plotting of traceplots did not highlight significant mixing issues.
In the real data experiments of Section \ref{sec:experiments}, four MCMC runs for each dataset were performed to reduce Monte Carlo error, see more details below.
MCMC runtimes were roughly 1 hour per run for Section \ref{sec:sim_study}, 20 hours per run for Section \ref{sec:SDS} and 50 hours per run for Section \ref{sec:SDS}.
The algorithms were implemented in \proglang{R} and a desktop computer with 32GB of RAM and an i9 Intel processor was used to perform the simulations.

When implementing the \emph{chaperones algorithm} of \citet[Appendix B]{miller15microclustering}, we used a non-uniform probability of selecting chaperones $i,j\in\{1,\dots,n\}$, assigning higher probability to pairs of records whose values agree on a large number of randomly selected fields
\footnote{The latter is done by first sampling a random number $N_f$ of fields between $0$ and $L$, then picking $N_f$ fields uniformly at random and then pick the chaperones $i,j\in\{1,\dots,n\}$ uniformly at random among those that agree on those $N_f$ fields. Other strategies could be used to favor pairs of chaperones that agree on various fields and we claim no optimality of this specific implementation.}.
This approach greatly improves convergence of the algorithm and respects the assumptions that the probability of selecting any pair of records is strictly greater than zero and is independent of the current partition, which are necessary to ensure the validity of the chaperones algorithm (see \citealp[Appendix B]{miller15microclustering}).
We expect the use of the chaperones algorithm with non-uniform proposals to be particularly beneficial in contexts with very small clusters, while for cases of larger clusters we expect the latter algorithm to behave similarly to standard split and merge schemes \citep{jain2004split}.

\begin{table}[h!]
\setlength{\tabcolsep}{0.5em}
{\renewcommand{\arraystretch}{1}
\begin{footnotesize}
\centering
\begin{tabular}{c||cc||cc}
 & \multicolumn{2}{c||}{SDS} & \multicolumn{2}{c}{SIPP}  \\ 
  \hline
   \hline
Model & FNR SE & FDR SE & FNR SE & FDR SE\\
\hline
\hline
DP & 0.03  &  0.04 & 0.02  & 0.01 \\
PY &  0.02  &  0.04  &  0.02  & 0.01\\
ESCNB & 0.02  &  0.04  &  0.01  & 0.02\\
ESCD & 0.02 & 0.02  &  0.01 & 0.01\\
\hline
\hline
\end{tabular}
\caption{Time-series MCMC error (in percentages) for the posterior expected values of FNR and FDR  for SDS and SIPP data sets.}\label{tab:mcmc_errors}
\end{footnotesize}
}
\end{table}

Figures \ref{fig:SDSconverge} and \ref{fig:SIPPconverge} show the traceplots for K, FNR and FDR for the four chains used for the SDS and SIPP data sets, respectively. No issues of convergence are observed in either case. However, the mixing of the chains for the SDS is slower compared to the SIPP data. 
Table \ref{tab:mcmc_errors} displays the estimated MCMC standard errors for the estimation of the average posterior FNR and FDR using the four chains and discarding the first 5,000 iterations of each run as a burn-in. The MCMC standard errors were computed using the function \emph{summary.mcmc} from the \proglang{R} package \pkg{CODA} \citep{CODA}. 
The estimated standard errors are all between $0.01\%$ and $0.04\%$, indicating that the FNR and FDR estimates presented in Section \ref{sec:experiments} of the main document are reliable up to one decimal place (in percentage), which is the level of precision reported in Tables 2 and 3 of the main document.

\begin{figure}[h] 
\includegraphics[width=0.32\textwidth]{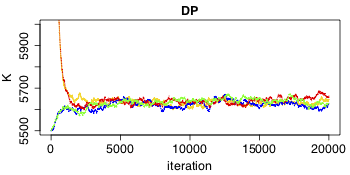}
\includegraphics[width=0.32\textwidth]{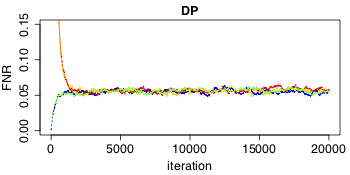}
\includegraphics[width=0.32\textwidth]{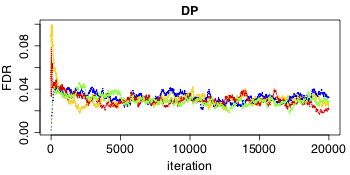}
\includegraphics[width=0.32\textwidth]{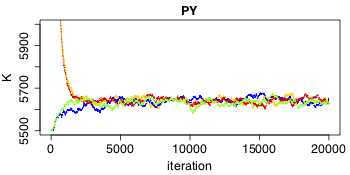}
\includegraphics[width=0.32\textwidth]{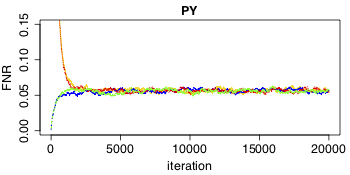}
\includegraphics[width=0.32\textwidth]{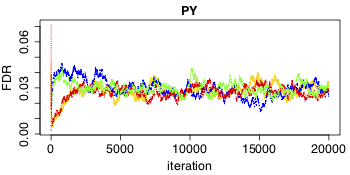}
\includegraphics[width=0.32\textwidth]{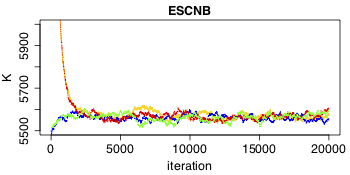}
\includegraphics[width=0.32\textwidth]{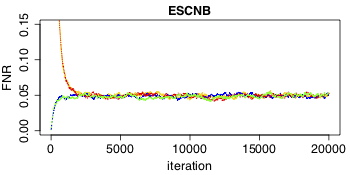}
\includegraphics[width=0.32\textwidth]{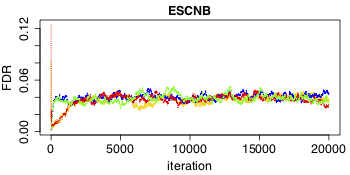}
\includegraphics[width=0.32\textwidth]{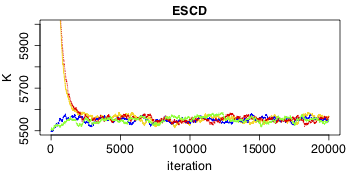}
\hspace{1mm}
\includegraphics[width=0.32\textwidth]{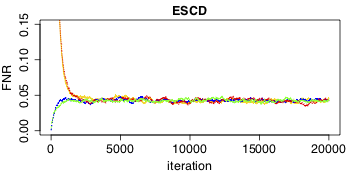}
\hspace{1mm}
\includegraphics[width=0.32\textwidth]{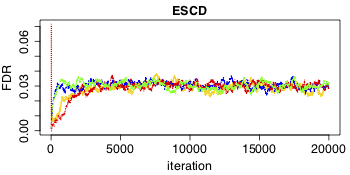}
\caption{\textbf{SDS dataset.}Trace plots of number of clusters (K), false negative rate (FNR) and false discovery rate (FDR) for four chains of 20,000 iterations of DP, PY, ESC-NB and ESC-D models for SDS data set of $K=5,500$.} \label{fig:SDSconverge}
\end{figure}

\begin{figure}[h] 
\includegraphics[width=0.32\textwidth]{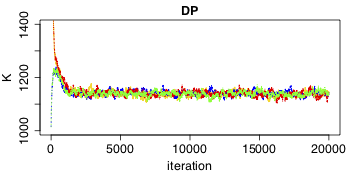}
\includegraphics[width=0.32\textwidth]{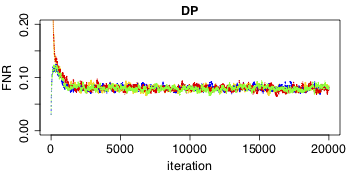}
\includegraphics[width=0.32\textwidth]{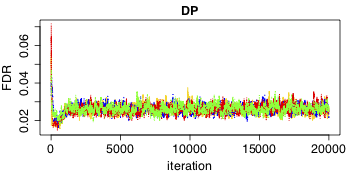}
\includegraphics[width=0.32\textwidth]{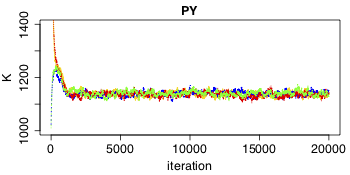}
\includegraphics[width=0.32\textwidth]{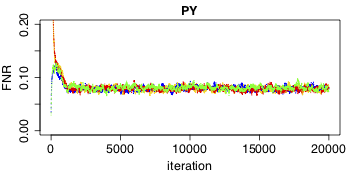}
\includegraphics[width=0.32\textwidth]{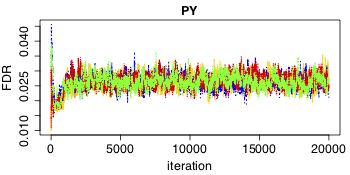}
\includegraphics[width=0.32\textwidth]{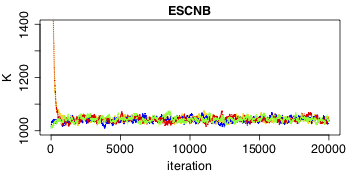}
\includegraphics[width=0.32\textwidth]{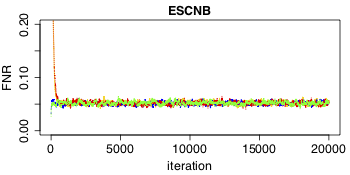}
\includegraphics[width=0.32\textwidth]{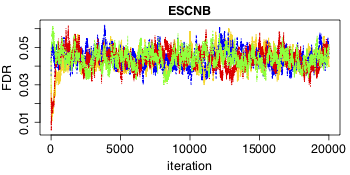}
\includegraphics[width=0.32\textwidth]{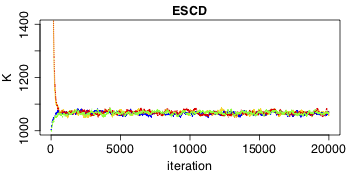}
\hspace{1mm}
\includegraphics[width=0.32\textwidth]{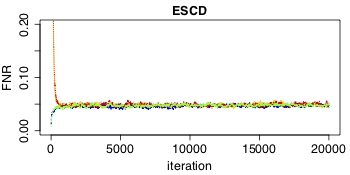}
\hspace{1mm}
\includegraphics[width=0.32\textwidth]{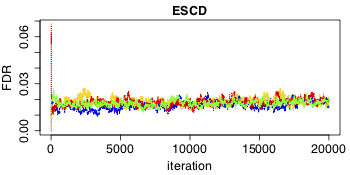}
\caption{\textbf{SIPP dataset.}Trace plots of number of clusters (K), false negative rate (FNR) and false discovery rate (FDR) for four chains of 20,000 iterations of DP, PY, ESC-NB and ESC-D models for SIPP data set of $K=1,000$.} \label{fig:SIPPconverge}
\end{figure}

\clearpage

\section{Additional results for the simulation study}\label{supp:sim_results}
\begin{figure}[h!]
\includegraphics[width=0.98\textwidth]{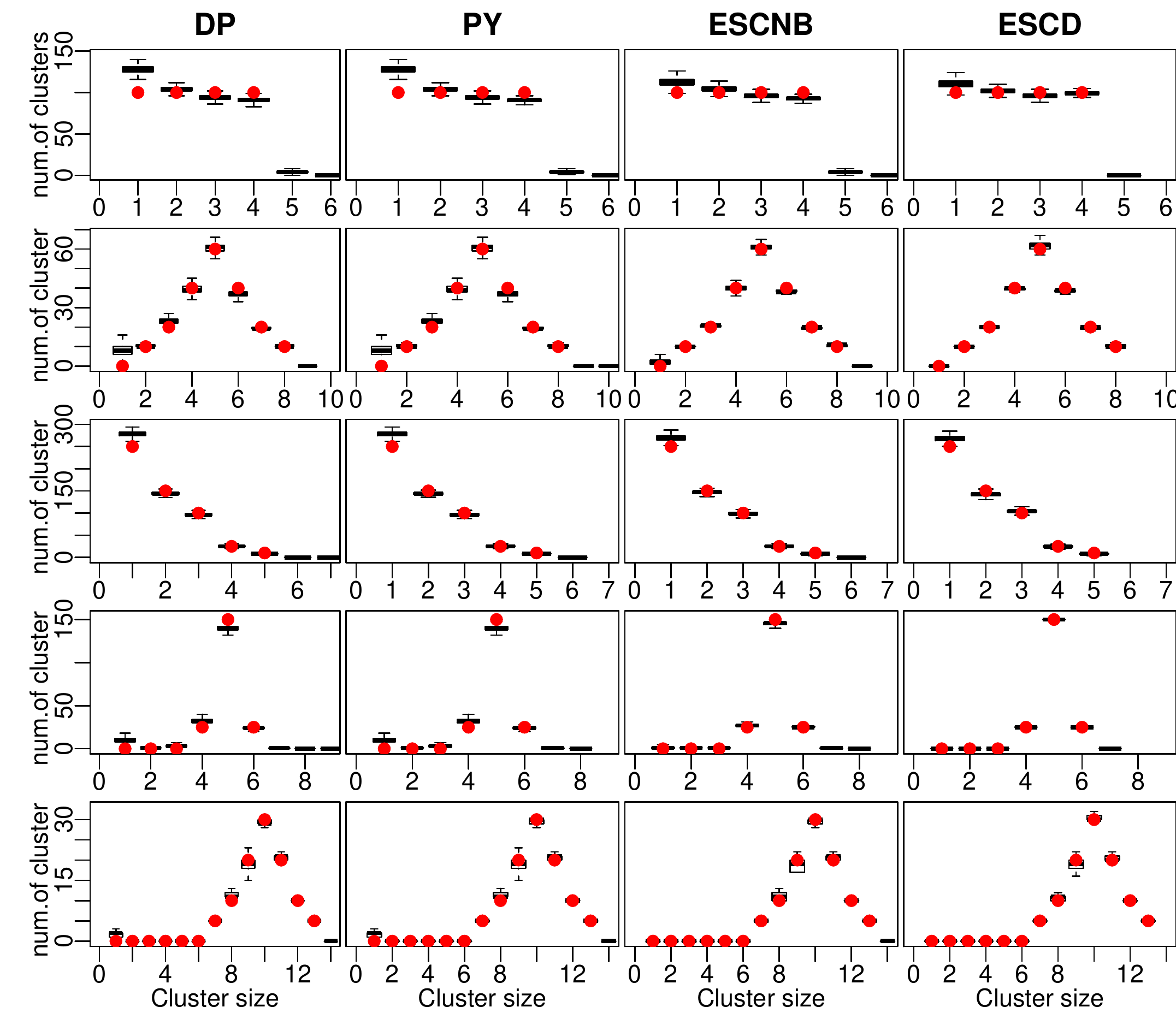}
\caption{Posterior distribution of the number of clusters of each size (black boxplots based on 20k MCMC samples from the posterior after thinning) versus number of clusters of each size in the true data-generating partition (red dots) for $\beta=0.01$. Each column corresponds to a different prior for the partition, and each row to a different data generating partition.} \label{fig:boxplots_simulations}
\end{figure}
\begin{figure}[h!]
\includegraphics[width=0.98\textwidth]{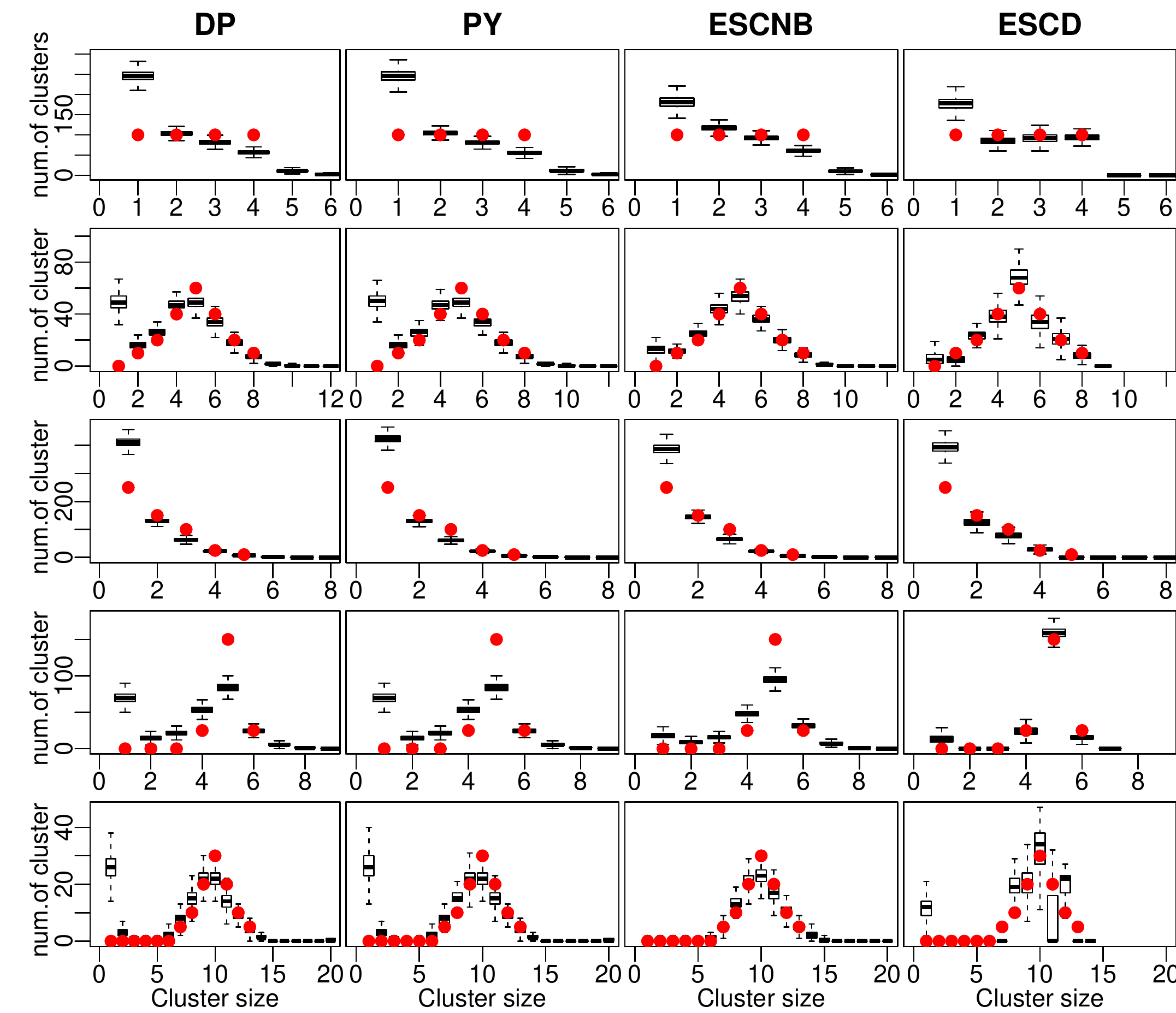}
\caption{Posterior distribution of the number of clusters of each size (black boxplots based on 20k MCMC samples from the posterior after thinning) versus number of clusters of each size in the true data-generating partition (red dots) for $\beta=0.10$. Each column corresponds to a different prior for the partition, and each row to a different data generating partition.} \label{fig:boxplots_simulations}
\end{figure}
\clearpage


\end{document}